\documentclass[11pt]{article}
\usepackage{amsfonts, amsmath, amssymb, amsthm}
\usepackage{bbm,cite, mathtools,xcolor}
\usepackage{comment,enumitem, eufrak}
\newtheorem{theorem}{Theorem}
\newtheorem*{theorem*}{Theorem}
\newtheorem{proposition}{Proposition}
\newtheorem{lemma}{Lemma}

\newtheorem{corollary}{Corollary}

\theoremstyle{definition}
\newtheorem{defn}{Definition}
\theoremstyle{remark}
\newtheorem*{remark}{Remark}

\usepackage{tikz}
\usepackage{pgfplots}
\usepackage[
font=small,labelfont=bf]{caption}
\usepackage{xcolor}

\newcommand{\eqdef}{\vcentcolon=}

\newcommand{\x}{\ensuremath{\mathbf{x}}}
\newcommand{\X}{\ensuremath{\mathbf{X}}}
\newcommand{\Y}{\ensuremath{\mathbf{Y}}}
\renewcommand{\d}{\ensuremath{\mathbf{d}}}
\renewcommand{\i}{\ensuremath{\mathbf{i}}}
\renewcommand{\j}{\ensuremath{\mathbf{j}}}
\newcommand{\y}{\ensuremath{\mathbf{y}}}

\renewcommand{\c}{\ensuremath{\mathbf{c}}}
\newcommand{\w}{\ensuremath{\mathbf{w}}}
\renewcommand{\v}{\ensuremath{\mathbf{v}}}

\renewcommand{\a}{\ensuremath{\mathbf{a}}}
\renewcommand{\b}{\ensuremath{\mathbf{b}}}
\newcommand{\0}{\ensuremath{\mathbf{0}}}

\newcommand{\F}{\ensuremath{\mathbb{F}}}
\newcommand{\cF}{\ensuremath{\mathcal{F}}}
\newcommand{\C}{\ensuremath{\mathcal{C}}}
\newcommand{\A}{\ensuremath{\mathfrak{A}}}
\renewcommand{\L}{\ensuremath{\mathcal{L}}}
\newcommand{\Z}{\ensuremath{\mathbb{Z}}}
\renewcommand{\epsilon}{\varepsilon}

\newcommand{\Mod}[1]{\ (\mathrm{mod}\ #1)}

\newcommand{\Modsp}[2]{\ (\mathrm{mod}_{#2}^*\ #1)}

\usepackage{array}
\newcolumntype{C}{>{$}c<{$}} %

\begin{document}
\title{Lifted Multiplicity Codes}

\date{}
\author{Lukas~Holzbaur\thanks{Technical University of Munich. Email: {\tt lukas.holzbaur@tum.de}. Supported by the Technical University of Munich -- Institute for Advanced Study, funded by the German Excellence Initiative and European Union 7th Framework Programme under Grant Agreement No. 291763 and the German Research Foundation under Grant No. WA3907/1-1.} \and Rina~Polyanskaya\thanks{Institute for Information Transmission Problems. Email: {\tt rev-rina@yandex.ru}. Supported in part by the Russian Foundation for Basic Research (RFBR) under Grant No.~\mbox{20-01-00559}.} \and Nikita Polyanskii\thanks{Technical University of Munich. Email: {\tt nikita.polyansky@gmail.com}. Supported in part by the German Israeli Project Cooperation (DIP) grant under Grant No.~KR3517/9-1.} \and Ilya Vorobyev\thanks{Skolkovo Institute of Science and Technology. Email: {\tt vorobyev.i.v@yandex.ru}. Supported in part by RFBR and JSPS under Grant No.~20-51-50007, and by RFBR under Grant No.~20-01-00559.} \and Eitan Yaakobi\thanks{Technion --- Israel Institute of Science and Technology. Email: {\tt yaakobi@cs.technion.ac.il}. Supported in part by the Israel Science Foundation  under Grant No. 1817/18 and by the Technion Hiroshi Fujiwara Cyber Security Research Center and the Israel National Cyber Directorate.} }

\maketitle
\begin{abstract}
Lifted Reed-Solomon codes and multiplicity codes are two classes of evaluation codes that allow for the design of high-rate codes that can recover every codeword or information symbol from many disjoint sets. Recently, the underlying approaches have been combined to construct lifted bi-variate multiplicity codes, that can further improve on the rate. We continue the study of these codes by providing lower bounds on the rate and distance for lifted multiplicity codes obtained from polynomials in an arbitrary number of variables.

 Specifically, we investigate a subcode of a lifted multiplicity code formed by the linear span of $m$-variate monomials whose restriction to an arbitrary line in $\F_q^m$ is equivalent to a low-degree uni-variate polynomial. We find the tight asymptotic behavior of the fraction of such monomials when the number of variables $m$ is fixed and the alphabet size $q=2^\ell$ is large.

  For some parameter regimes, lifted multiplicity codes  are then shown to have a better trade-off between redundancy and  the number of disjoint recovering sets for every codeword or information symbol than previously known constructions. Additionally, we present a local self-correction algorithm for lifted multiplicity codes.
	\end{abstract}

    \newpage
\section{Introduction}

The concepts of \emph{locality} and \emph{availability} of codes have been subject to intensive studies.  Informally, the locality of a code refers to the number of codeword symbols that needs to be accessed in order to recover a single codeword or information symbol and availability is the number of such (disjoint) recovery sets. These properties are of interest in a variety of applications, such as load balancing in distributed data storage, cryptography, and low-complexity error correction/detection. Several different notions related to these parameters have been considered in literature, including, but not limited to, locally recoverable codes (LRCs) \cite{huang2013pyramid,GHSY12}, locally decodable/correctable codes (LDCs/LCCs)~\cite{katz2000efficiency,yekhanin2012locally}, relaxed LCCs~\cite{gur2018relaxed} and LDCs~\cite{ben2006robust}, batch codes~\cite{holzbaur2020lifted,ishai2004batch}, PIR codes~\cite{fazeli2015pir}, and codes with the disjoint repair group property (DRGP)~\cite{li2019lifted}.

Reed-Muller (RM) codes are a popular class of codes that can provide strong locality and availability properties, as already exploited in the early majority-logic decoding algorithms \cite{reed1954class}.
These codes are defined as the evaluation of multi-variate polynomials up to a specific degree in all points of a multi-dimensional space. Their restriction to the evaluation points that fall on one line in this evaluation space can readily be seen to be equivalent to the evaluation of a uni-variate polynomial in the variable over the one-dimensional space spanned by this line. If the degree of this uni-variate polynomial is low, these positions form a codeword of a (non-trivial) Reed-Solomon (RS) code, another well-studied class of evaluation codes. This principle can be exploited to show locality and availability properties of the RM code. The locality properties of RM codes have been subject to extensive study (see, e.g., \cite{arora2003improved,alon2005testing,rubinfeld1996robust}).  However, the obvious drawback of RM codes with nice local recovery properties is their rather low rate of $R\leq 1/2$.

 To overcome this issue of low rate, the concept of \emph{lifted RS codes} was introduced in~\cite{guo2013new}. Instead of evaluating only multi-variate polynomials of a limited degree, as in RM codes, these codes consist of the evaluation of all polynomials that are equivalent to the evaluation of a low-degree uni-variate polynomial when restricted to a line. Using this concept of lifting, which first appeared in \cite{ben2011symmetric} in the context of LDPC codes, \cite{guo2013new} presents constructions of codes from multi-variate polynomials along with good bounds on the redundancy for the bi-variate case. These lead to codes of considerably higher rate than RM codes, which, broadly speaking, preserve the locality properties of the RM code. The main highlight of these codes is a construction of high-rate high-error LCCs. As a conceptual result, it was shown~\cite{guo2013new} that any polynomial producing a codeword of the lifted RS code can be decomposed to a linear combination of \textit{good} monomials whose restriction to lines are low-degree. Thus, the code rate is equal to the \textit{fraction} of good monomials. In \cite{holzbaur2020lifted}, for a fixed number of variables and large field size,  the asymptotic behaviour of this fraction was established. This improved on the estimate of rate of lifted RS codes for all cases with more than two variables. We remark that the distance properties of these codes follow from the fact that each symbol has many disjoint recovering sets and, thus, the relative distance of lifted RS codes is similar to the one of RM codes.

\emph{Multiplicity codes} \cite{kopparty2014high} are another recently introduced class of codes based on RM codes with good locality properties. Here, instead of each codeword symbol only consisting of the evaluation of a degree-restricted multi-variate polynomial, each symbol also contains the evaluation of all the derivatives of this polynomial up to some order. Similar to the concept of lifting, this generalization provides codes with significantly better rate than RM codes, while providing good locality properties. In particular, it was proved~\cite{kopparty2014high} that multiplicity codes represent a family of high-rate LCCs that have very efficient local decoding algorithms. The analysis of the rate of multiplicity codes is rather straightforward, whereas distance properties are implied by a bound on the number of points that a
low-degree polynomial can vanish on with high multiplicity.

As both lifted RS codes and multiplicity codes are based on generalizations of RM codes, it is a natural question whether these techniques can be combined to further improve the parameters of the respective codes. Some progress in the study of these \emph{lifted multiplicity codes} has recently been made in \cite{wu2015revisiting,li2019lifted}. In \cite{wu2015revisiting} the authors show asymptotic results for any number of variables. The focus of \cite{li2019lifted} is on improving the existence bounds on the required redundancy in the bi-variate case.

\subsection{Our contribution}

In this work we continue the study of lifted multiplicity codes by generalizing the results on the bi-variate case of \cite{li2019lifted} to an arbitrary number of variables. We investigate essentially the same class of codes as defined in~\cite{li2019lifted,wu2015revisiting}. Informally, the $[m,s,d,q]$ lifted multiplicity code consists of the evaluation (together with the derivatives up to the $s$th order) of polynomials from $\F_q[X_1,\ldots,X_m]$ whose restriction to a line agrees with some polynomial of degree less than $d$ on its first $s-1$ derivatives. Note that the condition $d<qs$ guarantees~\cite{wu2015revisiting,li2019lifted} that the all-zero codeword is produced only by the zero polynomial and, therefore, we fix $d=qs-r$ for some integer $r$.

Following a classic idea, we consider a subcode of a lifted multiplicity code formed by the linear span of \textit{good} monomials whose restriction to a line is equivalent to a low-degree polynomial. To count bad monomials, we first make use of our recent result~\cite{holzbaur2020lifted} for lifted RS codes $(s=1)$ and then extend it for larger $s$. Roughly speaking, we prove that there exists a one-to-$\binom{s+m-1}{m-1}$ correspondence between bad monomials for lifted RS codes and groups of  bad monomials for lifted multiplicity codes.  This enables us to find the exact asymptotic order of the number of bad monomials when $q$ is large (for more details, see Section~\ref{ss::rate and distance of lifted multiplicity codes}). Unfortunately, unlike lifted RS codes, there is no nice structural result saying that a good polynomial of a lifted multiplicity code can be decomposed into a linear combination of good monomials (for a counterexample see Section~ \ref{ss::equivalenceLiftedRS}). However, the fraction of good monomials serves as a lower bound on the rate of a lifted multiplicity code. Compared to prior works, our estimate is consistent with~\cite{li2019lifted} for $m=2$ and better than the result of~\cite{wu2015revisiting} for any $m\ge 2$.

Let $\binom{m}{\ge b}$ denote the number of ways to choose an (unordered) set of at least $b$ elements from a fixed set of size $m$. Our main contribution is summarized in the following statement.
\begin{theorem*}[Parameters of lifted multiplicity code]\ \\
\textbf{Code rate:} For powers of two $q$ and $s<q$ and a positive integer $r<q$, the rate of the $[m,s,qs-r,q]$ lifted multiplicity code is
$$
1 - O_m\left(s^{-1}(q/r)^{\log \lambda_m - m}\right)\quad \text{as } q\to\infty,
$$
where $\lambda_m$ is the largest eigenvalue of the matrix
$$
A_m \eqdef\left(\begin{smallmatrix}
\binom{m}{\ge 1} & \binom{m}{ 0} & 0 & 0 & \dots & 0 \\
\binom{m}{\ge 3} & \binom{m}{2} &\binom{m}{1} & \binom{m}{0} & \dots & 0 \\
\vdots & \vdots & \vdots & \vdots & \ddots & \vdots \\
\binom{m}{\ge 2j+1} & \binom{m}{ 2j} & \binom{m}{ 2j-1} & \binom{m}{ 2j-2} & \dots & \binom{m}{ 2j-m+2}
\\
\vdots & \vdots & \vdots & \vdots & \ddots & \vdots \\
\binom{m}{\ge 2m-1} & \binom{m}{2m-2} & \binom{m}{ 2m-3} & \binom{m}{ 2m-4} & \dots & \binom{m}{ m}
\end{smallmatrix}\right).
$$
\textbf{Distance:} For $r,s<q$, the relative distance $\Delta$ of the $[m,s,qs-r,q]$ lifted multiplicity code is
$$
\Delta\geq \Delta_{min}\coloneqq \left \lceil\frac{r-s+1}{s}\right\rceil\frac{q-s}{q^2}.
$$
For $s=o(q)$, $\Delta_{min}= \frac{r}{qs}(1+o(1))$.\\
\textbf{Availability:} Each symbol of a codeword of the $[m,s,qs-s,q]$ lifted multiplicity code can be reconstructed in $\lfloor q/s \rfloor^{m-1}$ different ways, each of which involves a disjoint set of coordinates of the codeword with cardinality $s^{m-1}(q-1)$.\\
\textbf{Local self-correction:} For $s^{m-2}=o(\log q)$ and $r<q$, let $\y$ be a noisy version of a codeword $\c$ of the $[m,s,qs-r,q]$ lifted multiplicity code such that the relative distance $\Delta(\y,\c)< \alpha\Delta_{min}$ with $0<\alpha<1/4$. Then for any $i\in[q^m]$, there exists a randomized algorithm $\A$ that makes at most $(q-1)s^{m-1}$ queries to $\y$ and reconstructs $c_i$ correctly with probability at least $1-2\alpha + o(1)$.

\end{theorem*}
The advantage of moving from lifted RS codes to lifted multiplicity codes is that the redundancy improves by a factor of $s$ (the order of derivatives), at the cost of the number of repair groups decreasing by a factor of $s^{m-1}$ and the logarithm of the alphabet size increasing by a factor of ${s + m - 1 \choose m}$. This means that lifted multiplicity codes cover more parameters of codes with good locality properties. For a relevant comparison, see the remarks after Lemmas~\ref{lem::known1}-\ref{lem::known2}.

Let us illustrate the improvement in the  rate of the $[m,s,qs-r,q]$ lifted multiplicity codes compared to the rate of the multiplicity code of order-$s$ evaluations of degree $qs-r$  polynomials in $m$ variables over $\F_q$~\cite[Lemma 7]{kopparty2014high}. Both types of codes have the same estimate on the relative distance $\Delta \ge \frac{r}{qs}(1+o(1))$. However, the rate of the multiplicity code is
$$
\frac{\binom{qs-r+m}{m}}{\binom{s+m-1}{m}q^m} < \left(\frac{qs-r+m}{(s+1/3) q}\right)^m \le 1 - \Omega_m\left(s^{-1}\right),
$$
which is smaller than the rate of lifted multiplicity codes as $\log \lambda_m< m$. Here, we point out that for large $m$, we are able to find the technical parameter $\lambda_m$ numerically only. We depict some values of $\lambda_m$ in Table~\ref{tab::eigenvalues}. This parameter stands for the exponential growth of the number of bad monomials. The inequality $\log \lambda_m< m$ follows from~\cite{guo2013new} implicitly, as the true exponent $\log \lambda_m$ was estimated by $m+\log \left(1-2^{-m\lceil \log m\rceil}\right)/\lceil \log m \rceil<m$. On the other hand, it is possible to estimate $\log \lambda_m$ from the other side by
\begin{equation}\label{eq::bounds on top eigenvalue}
\frac{-\log \left(1-2^{-m\lceil \log m\rceil}\right)}{\lceil \log m \rceil} \le m - \log \lambda_m \le -\log(1-2^{-m})
\end{equation}
and, thus, $m-\log \lambda_m >0$ vanishes as $m\to\infty$.
\begin{table}
	\centering
	\caption{The largest eigenvalue $\lambda_m$ of $A_m$, the resulting convergence rate $m-\log(\lambda_m)$ derived in~\cite{holzbaur2020lifted}, and the convergence rate $p_m$ of \cite{guo2013new} for different values of $m$.}
	\begin{tabular}{CCCC}
		m & \lambda_m & m-\log(\lambda_m) & p_m\\ \hline
		2 & 3.0000 & 4.1504 \times 10^{-1}& 4.1504 \times 10^{-1}\\
		3 & 7.2361 & 1.4479 \times 10^{-1}& 1.1360 \times 10^{-2}\\
		4 & 15.5436 & 4.1747 \times 10^{-2} & 2.8233 \times 10^{-3}\\
		5 & 31.7877 & 9.6043 \times 10^{-3} & 4.6986 \times 10^{-4}\\
		6 & 63.9217 & 1.7653 \times 10^{-3}& 1.1742 \times 10^{-4}\\
		7 & 127.9763 & 2.6714\times 10^{-4}& 2.9353\times 10^{-5} \\
		8 & 255.9939 & 3.4467 \times 10^{-5}& 2.8664 \times 10^{-8} \\
		9 & 511.9986 & 3.8959 \times 10^{-6} & 2.6872 \times 10^{-9}\\
		10 & 1023.9997 & 3.9323 \times 10^{-7}& 3.3590 \times 10^{-10}
	\end{tabular}
	\label{tab::eigenvalues}
\end{table}

Observe that if a good polynomial and its derivatives do not vanish on a point, then it can still be possible that the restrictions of the polynomial to some lines containing this point are equivalent to the zero polynomial. This fact was overlooked in~\cite{wu2015revisiting} when proving the distance property of lifted multiplicity codes. However, we can always say that the restriction of the polynomial to at least $(q-s)q^{m-2}$ lines crossing this point is equivalent to a non-zero uni-variate polynomial of degree less than $qs-r$ and, thus, the minimum distance of the code is at least $1+\lceil r/s - 1\rceil (q-s)q^{m-2}$ (for more details, see Section~\ref{ss::rate and distance of lifted multiplicity codes}).

Observe that the self-correction algorithm for multiplicity codes from~\cite{kopparty2014high} works well for lifted multiplicity codes. However, for small enough $s$, we present a slightly different local self-correction algorithm which requires $s\, 5^m$ times less locality. Here we combine two ideas: 1) for recovering of the evaluation of a polynomial and its derivatives up to the $s$th order at a point, it is sufficient to know directional derivatives for $s^{m-1}$ lines containing the point whose directional vectors $(1,v_2,\ldots,v_m)$ form a subcube $1\times Q_2\times\dots\times Q_m$ with $Q_i\subset \F_q$, $|Q_i|=p$; 2) every $(m-1)$-uniform hypergraph with $q$ vertices in each part with at least $\epsilon q^{m-1}$ hyperedges contains a copy of $(m-1)$-uniform clique with $s$ vertices in each part (for more details, see Section~\ref{ss::locally correctable codes}).

The availability property yields that lifted multiplicity codes have the best known trade-off between the number of information symbols $n$ and the required redundancy for private information retrieval (PIR) codes and codes with the disjoint repair group property (DRGP). The distinctive property of these codes is that every information (PIR code) or codeword (DRGP code) symbol can be recovered from $k$ disjoint subsets of codeword positions.  More precisely, from our results (for more details, see Section~\ref{ss::PIR codes from LMC}) it follows that given $n$, $m$, and $k=n^\epsilon$, with $0<\epsilon<1-1/m$, the required redundancy of non-binary and binary PIR codes constructed from $m$-variate multiplicity codes is $O(n^{\delta_{LM}(\epsilon,m)})$ and $O(n^{\delta'_{LM}(\epsilon,m)+o(1)})$, respectively, where
\begin{align*}
 \delta_{LM}(\epsilon,m)&:=\frac{m-1}{m} + \frac{1+\log \lambda_m - m}{m-1}\,\epsilon, \\
 \delta'_{LM}(\epsilon,m)&:=\frac{2m-1}{2m} + \frac{1+2\log \lambda_m - 2m}{2m-2}\,\epsilon.
\end{align*}
 We remark that for $m=2$ the same result was first derived in~\cite{li2019lifted}.

\subsection{PIR codes}
Now let us summarize the results for PIR codes, since the best known bounds for DRGP codes hold for PIR codes as well. The defining property of a $k$-PIR code is this: for every message symbol $x_i$, there exist $k$ disjoint sets of coded symbols from which $x_i$ can be uniquely recovered. Although this property is reminiscent of locally recoverable codes~\cite{GHSY12,TB14}, there are important differences. In locally recoverable codes, we wish to guarantee that every message symbol $x_i$ can be recovered from a \emph{small set} of coded symbols, and only one such recovery set is needed. Here, we wish to have \emph{many disjoint recovery sets} for every message symbol, and we do not care about their size.

Formally, this family of codes is defined as follows.

\begin{defn}[PIR code, %
	\cite{fazeli2015pir}]\label{def::PIR code}
	Let $F:\,\Sigma^n\to\Sigma^N$ be a map that encodes a string $x_1,\dots,x_n$ to $c_1,\dots, c_N$ and $\C$ be the image of $F$.
	The code $\C$ will be called a \textit{$k$-PIR code} (or \textit{$[N,n,k]_{|\Sigma|}^{P}$ code}) over the alphabet $\Sigma$ if for every $i\in[n]$, there exist $k$ mutually disjoint sets $R_{1},\dots , R_{k}\subset[N]$ (referred to as \textit{recovering sets}) and functions $g_1,\dots, g_k$ such that for all $\c\in\C$ and for all $j\in[k]$, $g_j(\c|_{R_{j}})=x_{i}$, where $\c|_{R}$ is the projection of $\c$ onto coordinates indexed by $R$.
\end{defn}

The main figure of merit when studying PIR codes is the value of $N$, given $n$ and $k$.  Denote by $N_q(n,k)$ the value of the smallest $N$ such that there exists an $[N,n,k]_q^P$ code. For the binary case, we will remove $q$ from these and subsequent notations. Since it is known that for sublinear $k$ and fixed $q$, $\lim\limits_{n\rightarrow \infty} N_q(n,k)/n=1$,~\cite{fazeli2015pir,guo2013new}, we evaluate these codes by their redundancy and define $r_q(n,k) := N_q(n,k)-n$. It is easy to see that for $k=2$, $r_q(n,2) = 1$, and for any fixed $k\ge 3$, $r_q(n,k) = \Theta(\sqrt{n})$~\cite{fazeli2015pir,rao2016lower,wootters2016linear}. %
In order to have a better understanding of the asymptotic behavior of the redundancy, the value of $r_q(n,k)$ is usually studied for $k=\Theta(n^\epsilon)$, $\epsilon\geq 0$.

The case of fixed $k$ was studied in~\cite{fazeli2015pir,VRK17}. There are several constructions of PIR codes~\cite{li2019lifted,asi2018nearly,FGW17,LC04,VRK17} and based on them, it is already possible to deduce some results on the asymptotic behavior of $r_q(n,k)$. For example, the constructions of \textit{one-step majority logic decodable codes} from~\cite{LC04} assure that $r(n,n^{\epsilon})= O(n^{0.5+\epsilon})$ for all $\epsilon\geq 0$. In~\cite{FGW17} the authors discussed partially lifted codes and their application to non-binary PIR codes. %
More results for PIR codes were achieved in~\cite{asi2018nearly} by using multiplicity codes and array codes. The recent construction~\cite{li2019lifted} of PIR codes is based on bi-variate lifted multiplicity codes. Constructions of PIR codes based on tri-variate lifted RS codes were investigated in~\cite{polyanskii2019lifted}. In Figure~\ref{fig::PIR}, we compare our results to the known results summarized in Lemma~\ref{lem::known1}-\ref{lem::known2}. It can be seen that for $1/2<\epsilon <1$, our bounds improve the state-of-art results.

\begin{lemma}\label{lem::known1}
The redundancy of non-binary PIR codes satisfies:
\begin{enumerate}
\item $r_q(n,k)= \Theta(\sqrt n)$ for fixed $k\ge 3$,~\cite{fazeli2015pir,rao2016lower,wootters2016linear}.
\item $r_q(n,n^{\epsilon})= O(n^{\delta(\epsilon)})$ for $0\leq \epsilon <1 $, where $\delta(\epsilon) = 1-\frac{1}{\lfloor 2/(1-\epsilon)\rfloor}+\frac{\epsilon}{\lfloor 2/(1-\epsilon)\rfloor-1}$,~\cite{asi2018nearly}.
\item $r_q(n,n^{0.25})= O(n^{0.714})$,~\cite{FGW17}.
\item $r_q(n,n^{\epsilon}) = O(n^{\frac{1}{2} + \epsilon (\log 3 - 1)})$ for $0\leq \epsilon < \frac{1}{2}, $~\cite{li2019lifted}.
\item $r_q(n,n^{1-1/m})=O(n^{1+\log \left(1-2^{-m\lceil \log m\rceil}\right)/(m\lceil \log m\rceil)})$ for an integer $m\ge 2$,~\cite{guo2013new}.
\item $r_q(n,n^{2/3})\le O(n^{\log_8(5+\sqrt{5})})$,~\cite{polyanskii2019lifted}.
\item $r_q(n,n^{1-1/m})=O(n^{\frac{\log \lambda_m}{m}})$ for an integer $m\ge 2$,~\cite{holzbaur2020lifted}.
\end{enumerate}
\end{lemma}
\begin{remark}
For easier comparison, we provide the relevant results for the best known families of non-binary PIR codes in the same form. For $0\le\epsilon\le (m-1)/m$, the required redundancy of $n^\epsilon$-PIR codes based on $m$-variate lifted multiplicity codes, $m$-variate multiplicity codes, and $m$-variate lifted RS codes is $O(n^{\delta_{LM}(\epsilon,m)})$, $O(n^{\delta_{M}(\epsilon,m)})$, and $O(n^{\delta_{LRS}(\epsilon,m)})$, respectively, where $\delta_{LM}(\epsilon,m) := \frac{m-1}{m} + \frac{1+\log \lambda_m - m}{m-1}\epsilon$, $\delta_M(\epsilon,m) := \frac{m-1}{m} + \frac{1}{m-1}\epsilon$ and $\delta_{LRS}(\epsilon,m):=\delta_{LM}(\frac{m-1}{m},m)$.
\end{remark}
\begin{lemma}\label{lem::known2}
The redundancy of binary PIR codes satisfies:
\begin{enumerate}
\item $r(n,k)= \Theta(\sqrt n)$ for fixed $k\ge 3$,~\cite{fazeli2015pir,rao2016lower,wootters2016linear}.
\item $r(n,n^{1-1/m})=O(n^{1+\log \left(1-2^{-m\lceil \log m\rceil}\right)/(m\lceil \log m\rceil)}\log n)$ for an integer $m\ge 2$,~\cite{guo2013new}.
\item $r(n,n^{\epsilon})= O(n^{0.5+\epsilon})$ for $0\leq \epsilon <1/2 $,~\cite{LC04, asi2018nearly}.
\item $r(n,n^{0.25})= O(n^{0.714}\log n)$,~\cite{FGW17}.
\item $r(n,n^{\epsilon})= O(n^{\delta(\epsilon)})$ for $0\leq \epsilon <1 $, where $\delta(\epsilon) =   \min\limits_{m\ge \lceil 1/(1-\epsilon)\rceil}\{1-\frac{m(1-\epsilon)-1}{2m(m-1)}\}$,~\cite{asi2018nearly}.
\item $r(n,n^{\epsilon}) = O(n^{\frac{3}{4} + \epsilon (\log 3 - \frac{3}{2})})$ for $0\leq \epsilon < \frac{1}{2}, $~\cite{li2019lifted}.
\item $r(n,n^{2/3})\le O(n^{\log_8(5+\sqrt{5})}\log n)$,~\cite{polyanskii2019lifted}.
\item $r(n,n^{1-1/m})=O(n^{\frac{\log \lambda_m}{m}}\log n)$ for an integer $m\ge 2$,~\cite{holzbaur2020lifted}.
\end{enumerate}
\end{lemma}

\begin{remark}
The codes constructed in \cite{FGW17,holzbaur2020lifted,guo2013new,polyanskii2019lifted} are $q$-ary codes of length $N=q^m$. To obtain a binary PIR code each symbol can be converted to $\log q = \log N^{\frac{1}{m}} = \frac{1}{m} \log N = \Theta (\log n)$ symbols, hence the additional factor of $\log(n)$ in Lemma~\ref{lem::known2} compared to Lemma~\ref{lem::known1}. Clearly, the image of every recovery set of a $q$-ary symbol is also a recovery set for bit of the image of this symbol, so the number of mutually disjoint recovering sets is at least as large as in for the non-binary code. We provide the relevant results for the best known families of binary PIR codes in the same form. For $0\le\epsilon\le (m-1)/m$, the required redundancy of binary $n^\epsilon$-PIR codes based on $m$-variate lifted multiplicity codes, $m$-variate multiplicity codes, and $m$-variate lifted RS codes is $O(n^{\delta_{LM}(\epsilon)+o(1)})$, $O(n^{\delta'_{M}(\epsilon)+o(1)})$, and $O(n^{\delta'_{LRS}(\epsilon)+o(1)})$, respectively, where $\delta'_{LM}(\epsilon,m) := \frac{2m-1}{2m} + \frac{1+2\log \lambda_m - 2m}{2m-2}\epsilon$, $\delta'_M(\epsilon,m) := \frac{2m-1}{2m} + \frac{1}{2m-2}\epsilon$ and $\delta'_{LRS}(\epsilon,m):=\delta'_{LM}(\frac{m-1}{m},m)$. Therefore, computing the bounds for small $m$ and employing the inequality~\eqref{eq::bounds on top eigenvalue} for large $m$, we can range these three families of binary $n^\epsilon$-PIR codes with  $\epsilon>2/3$ as follows
$$
\min_{m\ge \lceil \frac{1}{1-\epsilon}\rceil}\delta'_{LM}(\epsilon,m)< \min_{m\ge \lceil \frac{1}{1-\epsilon}\rceil}\delta'_{M}(\epsilon,m) < \min_{m\ge \lceil \frac{1}{1-\epsilon}\rceil}\delta'_{LRS}(\epsilon,m).
$$
\end{remark}

\begin{figure}[t]
  \centering
  \begin{tikzpicture}[thick,scale=1]
\pgfplotsset{compat = 1.3}
\begin{axis}[
	legend style={nodes={scale=0.7, transform shape}},
	legend cell align={left},
	width = 0.9\columnwidth,
	height = 0.63\columnwidth,
	xlabel = {$\log_n(k)$},
	xlabel style = {nodes={scale=0.8, transform shape}},
	ylabel = {$\log_n(r_q(n,k))$},
	ylabel style={nodes={scale=0.8, transform shape}},
	xmin = 0,
	xmax = 1,
	ymin = 0.48,
	ymax = 1.03,
	legend pos = south east]

\addplot [loosely dotted, very thick, color=red, mark=none] table[x=eps,y=bin] {upperBoundPIRFile.txt};
\addlegendentry{Upper bound (binary), Lemma~\ref{lem::known2}}%

\addplot[color= blue, mark=none,dotted,very thick] table[x=eps,y=binNew] {upperBoundPIRFile.txt};
\addlegendentry{Upper bound (binary), Theorem~\ref{th::binary disjoint repair group property code}}

\addplot[color=purple, mark=none,loosely dashed,very thick] table[x=eps,y=qary] {upperBoundPIRFile.txt};
\addlegendentry{Upper bound (non-binary), Lemma~\ref{lem::known1}}%

\addplot[color=green, mark=none,dashed, very thick] table[x=eps,y=qaryNew] {upperBoundPIRFile.txt};
\addlegendentry{Upper bound (non-binary), Theorem~\ref{th::asymptotic non-binary disjoint repair group code}}

\addplot[color=black,mark =none,thick] coordinates { (0,0.5) (0.5,0.5) (1,1) };
\addlegendentry{Lower bound~\cite{wootters2016linear,rao2016lower}}
\end{axis}
\end{tikzpicture}
  \caption{Comparison of parameters of binary and non-binary PIR codes based on lifted multiplicity  codes to the upper and lower bounds on the minimal redundancy of \cite{asi2018nearly,li2019lifted,wootters2016linear,rao2016lower,FGW17,polyanskii2019lifted}. For $\log_n(k)\leq 0.5$ the results of Theorem~\ref{th::binary disjoint repair group property code} and Theorem~\ref{th::asymptotic non-binary disjoint repair group code} recover the results from \cite{li2019lifted}.}
  \label{fig::PIR}
\end{figure}
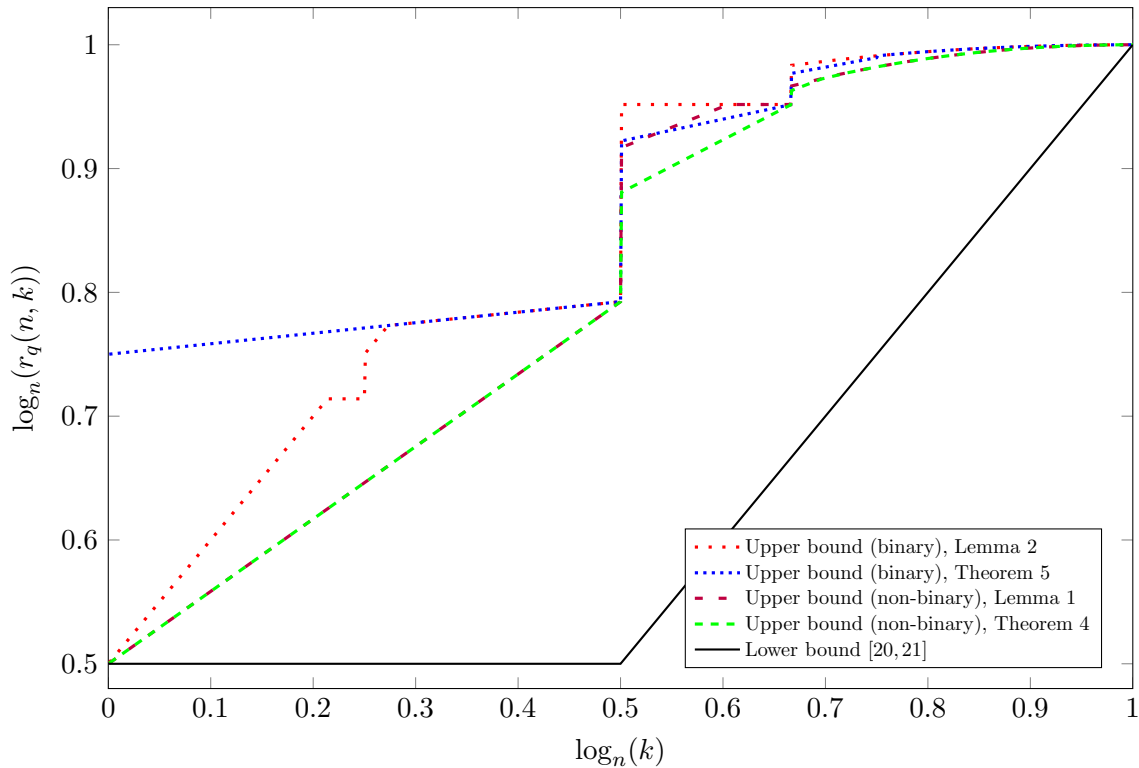

\subsection{Locally correctable codes}

Unlike PIR codes, LCCs \cite{katz2000efficiency} explicitly require locality properties. Informally, a code is said to be locally correctable if given a vector that is sufficiently close to a codeword, each codeword coordinate can be recovered from a small subset of (possibly noisy) other positions with high probability. We give a formal definition of LCCs below.
\begin{defn}[Locally correctable code.] \label{def::LCC}
    A code $\C$ of length $N$ over an alphabet $\Sigma$ is said to be $(r, \delta, \xi)$-locally correctable if there exists a randomized correcting algorithm $\A$ such that
    \begin{enumerate}
        \item For all $\c\in \C$, $i\in [N]$ and all vectors $\y \in \Sigma^N$ such that the relative distance $\Delta(\y, \c)\leq \delta $, we have $\Pr (\A(\y, i)=c_i)\geq 1-\xi$.
        \item $\A$ makes at most $r$ queries to $\y$.
    \end{enumerate}
\end{defn}
LDCs \cite{yekhanin2012locally} are defined similar to LCCs, except that there the algorithm is required to recover message symbols instead of codeword symbols. Note, that for linear codes local correctability is a strictly stronger notion than local decodability, as a systematically encoded LCC is always an LDC.

LCCs have been constructed employing different approaches such as RM codes, lifted RS codes~\cite{guo2013new}, multiplicity codes~\cite{kopparty2014high}, and tensor codes~\cite{ben2006LCC,viderman2015combination}.  One typical question about LCCs is phrased as follows: given the high rate of a code (close to 1), how to get the query complexity as small as possible. The current state-of-the-art construction provided in~\cite{kopparty2017high} has the sub-polynomial (in length) query complexity.
For an extensive discussion about other aspects of LCCs see~\cite{trevisan2004some, yekhanin2012locally, kopparty2017local} and the references therein.

\subsection{Outline}
The remainder of the paper is organized as follows. In Section~\ref{ss::prelimiminaries}, we give rigorous definitions of lifted multiplicity codes and introduce some auxiliary notation.
As the main result, bounds on the rate of lifted multiplicity codes and distance are derived in Section~\ref{ss::lifted mult codes}. In Section~\ref{ss::applications}, we apply these results to PIR codes and LCCs. Finally, we conclude with  Section~\ref{ss::conclusion}.

\section{Preliminaries}\label{ss::prelimiminaries}
We start by introducing some notation that is used throughout the paper.  For some function $f(x)$ and $g(x)$, we write $f(x)=O(g(x))$ and $f(x)=\Omega(g(x))$ as $x\to\infty$ if there exists some real $x_0$ and $C$ such that $|f(x)|\le C|g(x)|$ and $|f(x)|\ge C|g(x)|$ for $x\ge x_0$, respectively. If both equalities $f(x)=O(g(x))$ and $f(x)=\Omega(g(x))$ hold, then we use notation  $f(x)=\Theta(g(x))$. Also, we write $f(x)=o(g(x))$ as $x\to\infty$ if $|f(x)|\le \epsilon(x)|g(x)|$ for some function $\epsilon(x)$ such that $\epsilon(x)\to 0$ as $x\to\infty$. In these notations, we use a subscript, such as $O_m(f(x))$, if the parameter $m$ is supposed to be fixed.

Let $[n]$ be the set of integers from $1$ to $n$. We use uppercase letters such as $T$ and $X$ to denote variables. A vector is denoted by bold letters, e.g., $\d$ is a vector over a field or a ring and $\X$ is a vector of variables.  Let $q=2^{\ell}$ and $\F_q$ be a field of size $q$. We write $\log x$ to denote the logarithm of $x$ in base two. By $\Z_{\ge}$ and $\Z_{n}$  denote the set of non-negative integers and the set of integers from $0$ to $n-1$, respectively. In what follows, we fix $m$ to be a positive integer representing the number of variables. For $\d = (d_1,\dots, d_m)\in \Z_{q}^m$ and $\X=(X_1,\dots,X_m)$, let $\X^\d$ denote the monomial $\prod\limits{_{i=1}^m} X_i^{d_i}$ from $\F_q[\X]$. Let $\deg(\d)$ be the sum of components of $\d\in\Z_{\ge}^n$ and $|\d|$ be the number of non-zero components of $\d$. Additionally, we define $\deg_q(\d)\eqdef \sum_{i=1}^m \lfloor d_i / q \rfloor$. For a vector $\i \in \Z_{\ge}^m$, let $[\X^\i]f(\X)$ denote the coefficient of $\X^\i$ in the polynomial $f(\X)$. For $f(\X)\in\F_q[\X]$, we define $\deg(f)$ to be the maximal $\deg(\i)$ for $\i$ such that $[\X^\i]f(\X)$ is non-zero.

Let us define a partial order relation on $\Z_{q}$. For two integers $a=\sum_{i=0}^{\ell-1} a^{(i)} 2^i$ and $b=\sum_{i=0}^{\ell-1} b^{(i)} 2^i$ with $a^{(i)},b^{(i)}\in\{0,1\}$ we write $a\le_2 b$ 
if $a^{(i)}\le b^{(i)}$ for all $i\in\{0,\dots, \ell-1\}$. We denote $a=(a^{(\ell-1)},...,a^{(0)})_2$.
For vectors $\d, \d'\in \Z_{q}^m$, we write $\d\le_2 \d'$ if $d_i\le_2 d_i'$ for all $i\in[m]$. %

 Abbreviate the set of all lines in $\F_q^m$ by $\L_m\eqdef\left\{\w+\v T:\quad \w,\v\in\F_q^m \right\}$. For an $L=L(T)\in \L_m$ and a $f(\X)\in\F_q[\X]$, we write $f|_{L}$ to denote $f(L(T))$.
\subsection{Lifted multiplicity codes}
\begin{defn}
	For $f(\X)\in \F_q[\X]$ and a vector $\i\in\Z_{\ge}^m$, the $\i$th \textit{(Hasse) derivative} of $f$, denoted by $f^{(\i)}(\X)$, is the coefficient $[\Y^\i]g(\X,\Y)$, where the
	polynomial $g(\X,\Y):= f(\X+\Y)\in \F_q[\X,\Y]$. Therefore, we have
	$$
	g(\X,\Y) = \sum_{\i\in\Z_{\ge}^m} 
	f^{(\i)}(\X)\Y^\i.
	$$
\end{defn}
For an $\x\in \F_q^{m}$, an integer $s\ge1$
and a polynomial $f(\X) \in \F_q[\X]$, we write $f^{(<s)}(\x)\in \F_q^{\binom{s+m-1}{m}}$ to denote the
vector containing $f^{(\i)}(\x)$ for all $\i\in\Z_{\ge}^m$ so that $\deg(\i) < s$. In what follows, we assume that $s$ is a power of two.

We recall two well-known properties on Hasse derivates which will imply the linearity of lifted multiplicity codes over $\F_q$.
\begin{proposition}
	Let $f(\X),g(\X)\in \F_q[\X]$, $\lambda\in \F_q$ and let $\i\in \Z_{\ge}^{m}$. Then we have
	\begin{enumerate}
		\item  $f^{(\i)}(\X) + g^{(\i)}(\X) = (f + g)^{(\i)}(\X).$
		\item $(\lambda f)^{(\i)}(\X)=\lambda f^{(\i)}(\X).$
	\end{enumerate}
\end{proposition}
\begin{defn}\label{def::equivalenceOrderP}
	We say that two uni-variate polynomials $f(X),g(X)\in \F_q[X]$ are equivalent up
	to order $s$ if $f^{(<s)}(x) = g^{(<s)}(x)$ for all $x\in\F_q$. To indicate such equivalence, we write $f(X) \equiv_s g(X)$.
\end{defn}
The following statement shows the smallest possible degree of an equivalent polynomial. 
\begin{proposition}[Lemma 12 in~\cite{li2019lifted}]\label{pr::reducing the power}
	Let $q$ be a power of two. For every uni-variate polynomial $f(X)$,
	there exists a unique degree-at-most $sq-1$ polynomial $g(X)$ such that $f(X) \equiv_s g(X)$. Moreover, if $s$ is a power of two, then $f(X) = g(X) \pmod{X^{qs}+X^s}$ and for all $i$ such that $\deg(f) - qs +s < i < qs$, we have $[X^i]f(X)=[X^i]g(X)$.
\end{proposition}
Now we give a well-known result about multiplicities of a multi-variate polynomial.
\begin{lemma}[Follows from~\cite{dvir2013extensions}]\label{lem::multiplicity}
 Let $f(\X)$ be a non-zero polynomial of degree at most $d$. Then the number of points $\x\in\F_q^m$ such that $f^{(\i)}(\x)=0$ for all $\i\in\Z_{\ge}^m$ with $\deg(\i)<s$ is at most $\lfloor d q^{m-1} /s\rfloor$. 
\end{lemma}
 For a positive integer $d$, denote the set of uni-variate polynomials of degree less than $d$ by
$$
\cF_{d,q}\eqdef\{f(T)\in\F_q[T]:\,\,\deg(f)< d\}.
$$
\begin{defn}[Lifted multiplicity code~\cite{li2019lifted}]\label{def::lifted mult code}
        The  $[m,s,d,q]$ \textit{lifted multiplicity} code over $\F_q^{\binom{s+m-1}{m}}$ of length $q^m$ is defined as
        \begin{equation*}
          \mathcal{C} \eqdef \left\{
            \left.\left(f^{(<s)}(\w)\right)\right|_{\mathbf{w} \in\F_q^m} \ : \
        \begin{aligned}
          &f(\mathbf{X}) \in \F_q[\mathbf{X}]\ \text{such that} \\
          & f|_L \equiv_s g(T) \ \forall \ L=L(T) \in \mathcal{L}_m\\
          &\text{for some} \ g \in \mathcal{F}_{q,d}
\end{aligned} 
\right\} .
        \end{equation*}
\end{defn}
\begin{remark}
Multiplicity codes, as defined in \cite{kopparty2014high}, consist of the evaluations of multi-variate polynomials of degree $<d$. These polynomials trivially fulfill the condition that their restriction to every line $L \in \mathcal{L}_m$ is a polynomial of degree $<d$. It follows that the $[m,s,d,q]$ multiplicity code is a subcode of the $[m,s,d,q]$ lifted multiplicity code and thereby that the dimension of a lifted multiplicity code is lower bounded by the dimension of the corresponding multiplicity code. However, for many parameters, lifting increases the rate of the multiplicity code, as we formally show in Section~\ref{ss::lifted mult codes}. To provide some further intuition, we also give an example for this improvement in Appendix~\ref{ss::improvement LMC vs MC}.
\end{remark}
Define an operation $\Modsp{q}{s}$ that takes a non-negative integer and maps it to the element from $\Z_{qs}$ as follows
$$
a \Modsp{q}{s}\eqdef
\begin{cases}
a,\,&\text{if   } a\in\Z_s,\\
b\in \Z_{qs}\setminus \Z_s,\,&\text{if }a\not\in\Z_s,\,a=b \Mod{qs-s}.
\end{cases}
$$
It can be easily checked that if $a\,\Modsp{q}{s}=b$, then $T^{a} \equiv_s T^{b}$.
\begin{defn}[$(d,s)$-bad and good monomials] \label{def::bad (d^*,s) monomial}
	Given positive integers $s$ and $d$, we say that a monomial $\X^\d$ with $\d\in\Z_{qs}^m$ and $\deg_q(\d)\le s-1$ is \textit{$(d,s)$-bad} over $\F_q[\X]$ if there exists at least one $\i\in\Z_{qs}^m$ such that $\i\le_2\d$ and $\deg(\i) \Modsp{q}{s} \in \{d,d+1,\dots,qs-1\}$. A monomial $\X^\d$ with $\d\in\Z_{qs}^m$ and $\deg_q(\d)\le s-1$  is said to be \textit{$(d,s)$-good} if it is not $(d,s)$-bad.
\end{defn}
Let $\mathcal{F}_{q,s,d}$ be the collection of $(d,s)$-good monomials from $\F_q[\X]$.

\begin{proposition}\label{prop::code cardinality}
For $s\le q$, the cardinality of the $[m,s,d,q]$ lifted multiplicity code is  $\geq q^{|\mathcal{F}_{q,s,d}|}$.
\end{proposition}
\begin{proof}
The full proof of this technical statement is given in Appendix~\ref{ss::injective map}. There we show that different linear combinations of good monomials produce different codewords and that these codewords are contained in the $[m,s,d,q]$ lifted multiplicity code. Thus, the lower bound on the dimension of the code follows direclty from the number of good monomials $|\mathcal{F}_{q,s,d}|$.
\end{proof}
\begin{remark}
Observe that for $s=1$, Definition~\ref{def::lifted mult code} gives exactly the code spanned by the evaluation of good monomials, i.e., the statement of Proposition~\ref{prop::code cardinality} holds with equality. This case corresponds to lifted RS codes, for which this equivalence first appeared in~\cite{guo2013new}. Therefore, the $[m,1,d,q]$ lifted multiplicity code will be called the $[m,d,q]$ lifted RS code in the following. 

In Appendix~\ref{ss::equivalenceLiftedRS}, we provide some codewords of a lifted multiplicity code with $s\geq 2$, which are not included in the subcode spanned by the evaluation of monomials, thereby showing that the statement of Proposition~\ref{prop::code cardinality} does not hold with equality in general. 
\end{remark}

\section{Code rate and distance of lifted multiplicity codes} \label{ss::lifted mult codes}

In this section, as a warm-up, we first recall some known results for lifted RS codes corresponding to the case $s=1$. Then we investigate the code rate and the minimal distance of lifted multiplicity codes. We impose the constraint $s\ge m$ on the parameters, which helps with dropping the modulo operation in the definition of bad monomials. Then by applying the known results for lifted RS codes, we show how to find the asymptotics of the number of bad monomials when $m$ is fixed and $q$ is large. Our estimate continues the study of two-dimensional lifts initiated in~\cite{li2019lifted} and is consistent with the result with the result presented for the case of $m=2$ presented there.  
\subsection{Lifted Reed-Solomon codes}
We now recall a known estimate for the number of $(q-r,1)$-bad monomials when the number of variables is fixed and the alphabet size is large.
\begin{proposition}[Corollary 1 and 2 from \cite{holzbaur2020lifted}] \label{prop:: number of bad * monomials}
	For an integer $r<q=2^{\ell}$, the number of $(q-r,1)$-bad monomials is $\Theta_m\left(r^{m-\log\lambda_m} q^{\log \lambda_m}\right)$ as $\ell\to\infty$,
	where $\lambda_m$ is the largest eigenvalue of the matrix
	$$
	A_m \eqdef\left(\begin{smallmatrix}
	\binom{m}{\ge 1} & \binom{m}{ 0} & 0 & 0 & \dots & 0 \\
	\binom{m}{\ge 3} & \binom{m}{2} &\binom{m}{1} & \binom{m}{0} & \dots & 0 \\
	\vdots & \vdots & \vdots & \vdots & \ddots & \vdots \\
	\binom{m}{\ge 2j+1} & \binom{m}{ 2j} & \binom{m}{ 2j-1} & \binom{m}{ 2j-2} & \dots & \binom{m}{ 2j-m+2}
	\\
	\vdots & \vdots & \vdots & \vdots & \ddots & \vdots \\
	\binom{m}{\ge 2m-1} & \binom{m}{2m-2} & \binom{m}{ 2m-3} & \binom{m}{ 2m-4} & \dots & \binom{m}{ m}
	\end{smallmatrix}\right).
	$$	
Moreover, the number of $\d\in\Z_q^m$ such that there exists an $\i\in\Z_q^m$ with $\i\le_2 \d$ and 
\begin{enumerate}
	\item $\deg(\i) \pmod{q}\in\{q-r,q-r+1,\ldots,q-1\}$ is  $\Theta_m\left(r^{m-\log\lambda_m} q^{\log \lambda_m}\right)$ as $\ell\to\infty$.
	\item $\deg(\i)\in\{q-r,q-r+1,\ldots,q-1\}$ is also $\Theta_m\left(r^{m-\log\lambda_m} q^{\log \lambda_m}\right)$ as $\ell\to\infty$.
\end{enumerate}
\end{proposition}

Next we show how to derive the code rate and the minimal distance of lifted RS codes.
\begin{theorem}[{Rate and distance of lifted RS codes, \cite[Theorem 1]{holzbaur2020lifted}}]\label{th:: number of bad monomials}
	For a power of two $q$, the rate $R$ and the relative distance $\delta$ of the $[m,q-r,q]$ lifted RS code are
	$$
	R = 1-\Theta_m\left((q/r)^{\log \lambda_m-m}\right), \quad \delta \ge \frac{r}{q}\quad \text{as } q\to\infty.
	$$
\end{theorem}
\begin{remark}
It is clear that for any $\epsilon>0$, there exist some real $c>0$ such that for $r=cq$, the rate $R\ge 1-\epsilon$ and the relative distance $\delta\ge c$. Also, it can be seen that rate approaches $1$ for $r=o(q)$ as $\lambda_m<2^m$. These facts were also proved in~\cite{guo2013new} in order to show the existence of high rate high error locally correctable codes and high rate LCCs with sublinear locality. Let us illustrate the improvement of Theorem~\ref{th:: number of bad monomials} compared to the result from~\cite{guo2013new}. We take $r=O(1)$ and see that the convergence rate of our estimate is $1-\Theta_{m,r}\left(q^{\log \lambda_m - m}\right)$. 
The arguments from~\cite{guo2013new} show that for $m\ge 2$, the rate is
$$
1- O_{m,r}\left(\left(1-2^{-m\lceil\log m\rceil}\right)^{\log q/\lceil\log m\rceil}\right)= 1-O_{m,r}(q^{-p_{m}}),
$$
where $p_m\eqdef -\log\left(1-2^{-m\lceil\log m\rceil}\right)/\lceil\log m\rceil$. In Table~\ref{tab::eigenvalues}, we depict some values of $m-\log \lambda_m$ and $p_m$ for $2\le m\le 9$.
\end{remark} 
\begin{proof}%
	To estimate the code rate of $[m,q-r,q]$ lifted RS codes, it suffices to compute the fraction of $(q-r,1)$-good monomials. By Proposition~\ref{prop:: number of bad * monomials}, the rate is
	$$
	1 - \Theta_m\left(r^{m-\log\lambda_m} q^{\log \lambda_m}\right)q^{-m}=1-\Theta_m\left((q/r)^{\log \lambda_m-m}\right)
	$$
	as $q\to\infty$. To estimate the relative distance of the code, we first note that the lifted RS code is linear. Suppose that $(f(\a))|_{a\in\F_q^m}$ is a non-zero codeword. Let us say that $f(\w_0)\neq 0$. Then for any $\v\in \F_q^m\setminus \{\0\}$, the polynomial $f(\w_0+\v T)$ is equivalent to a non-zero uni-variate polynomial of degree at most $q-r-1$. Thus, $f(\w_0+\v t)\neq 0$ for at least $r+1$ different values $t\in\F_q$ and $f(\a)$ is non-zero for at least  $1+rq^{m-1}$ values $\a\in\F_q^m$.  This completes the proof.
\end{proof}
\subsection{Computing  the number of $(qs-r,s)$-bad monomials} \label{ss::computing the number of bad monomials}
In this section, we show that the number of $(qs-r,s)$-bad monomials can be well approximated  by ``$\binom{s+m}{m-1}$  times the number of $(q-r,1)$-bad monomials''.  

Let $s\ge m$ be a power of two and $1\le r< q$. First, we show that for such a choice of parameters, the modulo operation in Definition~\ref{def::bad (d^*,s) monomial} can be dropped. By Proposition~\ref{pr::reducing the power}, for $f(X)\in\F_q[X]$ with $\deg(f)\le (s-1)q + m(q-1)=(m+s-1)q-m$, we have that $[X^i](f(X) \pmod{X^{qs}+X^{s}})=[X^i]f(X)$ for all $i\in \{qs-r, qs-r+1,\ldots, qs-1\}$ as $(m+s-1)q-m-qs+s=(m-1)q - m +s < qs - r$. Therefore, by Definition~\ref{def::bad (d^*,s) monomial}, a monomial $\X^{\d}$ with $\d\in\Z_{qs}^m$ and $\deg_q(\d)\le s-1$ is  $(qs-r,s)$-bad if there exists a vector $\i$ such that $\i\le_2 \d$ and $\deg(\i)\in \{qs-r,qs-r+1,\ldots, qs-1\}$.

Let a monomial $\X^{\d}$ be $(qs-r,s)$-bad. Then every component of $\d$ can be represented as $d_j = \hat d_j q+ d_j'$ with $d_j'\in \Z_q$ and $\hat d_j \in \Z_s$ for all $j\in[m]$. As deduced above, there exists an $\i\in\Z_{qs}^m$ such that $\i\le_2 \d$ and $\deg(\i)\in \{qs-r,qs-r+1,\ldots, qs-1\}$. Therefore, after representing $i_j=\hat i_j q + i_j'$, we obtain that $\i'\le_2 \d'$ and $\deg(\i') \pmod{q} \in  \{q-r,q-r+1,\ldots,q-1\}$.%
Let us also check that $s-m\le \deg(\hat \d)\le s-1$. To show $\deg(\hat \d)\ge s-m$, we just note that 
$$
\deg(\i)\le \deg(\d)= \deg(\hat \d)q + \deg(\d')\le \deg(\hat \d)q + (q-1)m.
$$
Thus, if  $\deg(\hat \d)< s-m$, we have that $\deg(\i)\le (s-1)q - m < qs - r$ which contradicts the property  $\deg(\i)\in \{qs-r,qs-r+1,\ldots, qs-1\}$. Note that $\deg(\hat \d)=\deg_q(\d)$, therefore $\deg(\hat \d)\le s-1$.
Finally, we arrive at the following statement.
\begin{corollary}\label{cor:: number of bad (qs-r,s) monomials}
	For an integer $m<r<q=2^{\ell}$ and a power of two $s\ge m$, the number of $(qs-r, s)$-bad monomials is 
	$$
	\Theta_m\left(s^{m-1}r^{m-\log \lambda_m} q^{\log \lambda_m}\right)\quad \text{as }\ell\to\infty.
	$$
\end{corollary}
\begin{proof}%
	As noted above, for every $(qs-r,s)$-bad monomial $\X^\d$, $\d$ can be uniquely decomposed to the pair $(\hat\d,\d')$, where $s-m\le \deg(\hat\d) \le s-1$ and for $\d'\in\Z_q^{m}$, there exists an $\i'\le_2 \d'$ with $\deg(\i') \pmod{q} \in  \{q-r,q-r+1,\ldots,q-1\}$. Thus, Proposition~\ref{prop:: number of bad * monomials} yields that the number of $(qs-r,s)$-bad monomials for $\ell\to\infty$ can be bounded by
	$$
	\left(\sum_{j=1}^{m}\binom{s-j+m-1}{m-1}\right) O_m\left(r^{m-\log \lambda_m} q^{\log \lambda_m}\right)  	=O_m \left(s^{m-1}r^{m-\log \lambda_m} q^{\log \lambda_m}\right).
	$$
	It remains to show that this estimate is asymptotically tight.
	To see this, consider all possible $\d'\in\Z_q^m$ such that there exists $\i'\in\Z_q^m$ with $\i'\le_2 \d'$ and $\deg(\i')=q-r'\in\{q-r,q-r+1,\ldots,q-1\}$. By Proposition~\ref{prop:: number of bad * monomials} the number of such $\d'$ can be estimated as
	$$
	\Omega_m\left(r^{m-\log \lambda_m} q^{\log \lambda_m}\right).
	$$
 Now we take a look on all possible $\hat\d\in \Z_s^m$ such that $\deg(\hat \d) = s-1$. We can estimate the number of such $\hat\d$ by 
	$\binom{s+m-2}{m-1}$. For any such $
	\hat \d$, we define $\d\in\Z_{qs}^m$ to be such that $d_j=\hat d_jq + d_j'$ and note that $\X^\d$ is $(qs-r,s)$-bad as for $\i$ with $i_j = \hat d_jq+i_j'$, we have $\i\le_2 \d$ and 
	$$
	\deg(\i)=q\deg(\hat \d) + \deg(\i') = q(s-1) + q-r' = qs - r',
$$
	which belongs to $\{qs-r,qs-r+1,\ldots, qs-1\}$. Therefore, the number of $(qs-r,s)$-bad monomials is 
	$$
	\binom{s+m-2}{m-1}\Omega_m\left(r^{m-\log \lambda_m} q^{\log \lambda_m}\right)
	= \Omega_m \left(s^{m-1}r^{m-\log \lambda_m} q^{\log \lambda_m}\right).
	$$
	This completes the proof.
\end{proof}
\subsection{Rate and distance of lifted multiplicity codes}\label{ss::rate and distance of lifted multiplicity codes}

\begin{theorem}[Rate and distance of lifted multiplicity codes]\label{th:: code rate of lifted multiplicity codes}
	For powers of two $s,q$ and integers $r$ and $m$ with $m\le s \le q$ and $r\le q$, the rate of the $[m,s,qs-r,q]$ lifted multiplicity code is
	$$
1 - O_m\left(s^{-1}(q/r)^{\log \lambda_m - m}\right)\quad \text{as } q\to\infty.
	$$
	The relative distance $\Delta$ of the $[m,s,qs-r,q]$ lifted multiplicity code is
$$
\Delta\geq \Delta_{min}:=\left \lceil\frac{r-s+1}{s}\right\rceil\frac{q-s}{q^2}.
	$$
	For $s=o(q)$,  $\Delta_{min}= \frac{r}{qs}(1+o(1))$.
\end{theorem}
\begin{proof}[Proof of Theorem~\ref{th:: code rate of lifted multiplicity codes}]
	By Proposition~\ref{prop::code cardinality}, we can obtain the lower bound on the rate of the lifted multiplicity  code by computing the fraction of $(qs-r,s)$-good monomials. Thus, by Corollary~\ref{cor:: number of bad (qs-r,s) monomials}, the rate is 
	$$
	1- \frac{O_m\left(s^{m-1}r^{m-\log \lambda_m} q^{\log \lambda_m}\right)}{\binom{s+m-1}{m}q^m} = 1 - O_m\left(s^{-1}(q/r)^{\log \lambda_m - m}\right).
	$$

	Now we estimate the distance of the $[m,s,qs-r,q]$ lifted multiplicity code. Consider a codeword which is the evaluation of some non-zero polynomial $f$.
	Let $\w_0\in\F_q^m$ be a coordinate such that $f^{(<s)}(\w_0)$ is not all-zero. In what follows, we prove the existence of a set $S$, $|S|\geq (q-s)q^{m-1}$, of lines containing this point such that for any $L\in S$ polynomial $f|_L$  doesn't vanish for at least $\lceil r/s\rceil$ points.  More explicitly, assume that for some $\i_0\in\Z_{\ge}^m$ with $\deg( \i_0)=i_0 < p$, $f^{(\i_0)}(\w_0)\neq 0$. Let a line $L$ be parameterized by $\w_0+T\v$ with $\v=(1,v_2,\ldots,v_m)$, $v_i\in\F_q$. Define $g_\v(T):=f|_L=f(\w_0+T\v)$. By the definition of Hasse derivatives, we have
	$$
	g_\v(T) = \sum_{\i\in\Z_{\ge}^m}f^{(\i)}(\w_0+T\v) T^{\deg(\i)} \v^\i
	$$
	and, thus,
	$$
		g_{\v}^{(i_0)}(0)=\sum\limits_{\i :\ \deg(\i)=i_0}f^{(\i)}(\w_0) \v^{\i}.
	$$
	Since $f^{(\i_0)}(\w_0)\neq 0$, we can think about the right-hand side of the above equality as a non-zero polynomial in $v_2,\ldots,v_m$ of degree at most $s$. This yields that there exist at most $s q^{m-2}$ different $\v=(1,v_2,\ldots, v_m)\in\F_q^m$ such that $g_\v^{(i_0)}(0)=0$. Thus, for at least $(q-s)q^{m-2}$ different lines $L$ containing the point $\w_0$, the uni-variate polynomial $f|_L\neq 0$. By the definition of $[m,s,qs-r,q]$ lifted multiplicity codes, for any line $L$, $f|_L$ agrees with some uni-variate polynomial of degree at most $qs-r-1$ on its first $s-1$ derivatives. By Lemma~\ref{lem::multiplicity}, if $g_{\v}(T)=f|_L\neq 0$, there exist at least $\lceil (r+1)/s\rceil$ points on which $f|_L$  doesn't vanish with high multiplicity, i.e., for at least $\lceil (r+1)/s\rceil$ different $t\in\F_q$, $g_\v^{(j)}(t)\neq 0$ for some $j<s$. This implies that the number of non-zero positions of the codeword produced by $f$ is at least 
	$$
	1+\left\lceil \frac{r+1}{s}-1\right\rceil (q-s)q^{m-2}.
	$$
	Since the lifted multiplicity code is $\F_q$-linear, the distance of the lifted multiplicity code can be bounded by the same value. This completes the proof.
\end{proof}

\section{Applications}\label{ss::applications}
\subsection{PIR codes from lifted multiplicity codes}\label{ss::PIR codes from LMC}
In the previous sections, we derived bounds on the rate of lifted multiplicity codes, which we use in this section to obtain new upper bounds on the required redundancy of PIR codes (cf. Definition~\ref{def::PIR code}).
Our results improve the constructions of these codes based on ordinary multiplicity codes~\cite{asi2018nearly}.
Note that the definition of a \textit{code with the disjoint repair group property} (DRGP)~\cite{li2019lifted} is similar to Definition~\ref{def::PIR code}, except that we should recover all codeword symbols instead of only information symbols. For $\F_q$-linear codes, as in our case, any systematically encoded code with the DGRP directly gives a PIR code. The codes constructed from lifted multiplicity codes in the following have the DGRP property, but as the focus here are PIR codes, we state the results for this code class.

First let us recall a known result for recovering the evaluation $f^{(<s)}(\w_0)$ for an arbitrary polynomial. 

\begin{lemma}[Follows from~{\cite[Theorem~14]{asi2018nearly}}]\label{lem::good direction and multivariate polynomial}
	 Let $f(\X)\in\F_q[\X]$ and a line $L$ be parameterized as $\w_0+T\v$. Define $g_\v(T):=f|_L=f(\w_0+T\v)$. Let a family of sets $Q_2, \ldots, Q_{m}$, $Q_i\subset \F_q$, $|Q_i|=s$, be given. If for all directions of the form $\v=(1, v_2, \ldots, v_{m})$, $v_i\in Q_i$, and all $0\leq j< s$, values $g_{\v}^{(j)}(0)$ are known,  then it is possible to reconstruct $f^{(<s)}(\w_0)$.
\end{lemma} 

Next we prove that lifted multiplicity codes satisfy the definition of $k$-PIR codes for appropriate $k$. 
\begin{theorem}[Lifted multiplicity codes are PIR codes]\label{th:: nonbinary disjoint repair group code}
	Fix an integer $m\ge 2$ and powers of two $q$ and $s$ with $m\le s \le q$. The $[m,s,qs-s,q]$ lifted multiplicity code is a $k$-PIR code for $k=(q/s)^{m-1}$.
\end{theorem}
\begin{proof}%
For any line $L$ parameterized by $\w_0+T\v$ and a polynomial $f$ producing a codeword of the $[m,s,qs-s,q]$ lifted multiplicity code, the polynomial $g_\v(T):=f|_L$ is equivalent up to order to $s$ to a uni-variate polynomial $h(T)$ of degree at most $sq-s-1$. By reading $g_\v^{(j)}(t)$ for all $0\le j<s$, $t\in\F_q\setminus\{0\}$, we can reconstruct polynomial $h(T)$ in $O(qs\log(qs))$ time (cf.~\cite{chin1976generalized}) and get the values $h^{(j)}(0)=g_\v^{(j)}(0)$ for all $0\le j<s$.

For an integer $i\in [q/s]$, let $Q_i$ be a subset of $\F_q$ of size $s$ so that $Q_i\cap Q_j=\emptyset$ for $j\neq i$. Let us index codeword symbols by elements of $\F_q^m$, i.e., $(c_1,\ldots,c_{q^m})=(c_{\w})|_{\w\in \F_q^m}$, where $c_{\w}=f^{(<s)}(\w)$. Fix an arbitrary vector $(i_2,\ldots, i_{m})\in [q/s]^{m-1}$. By Lemma~\ref{lem::good direction and multivariate polynomial}, for $\w_0\in\F_q^m$, a possible recovering set for $c_{\w_0}$ is simply 
	$$
	\left\{\w_0+\v t:\, t\in \F_q\setminus\{0\},\,v_1=1,\,v_j\in Q_{i_j}\text{ for }j\in[m]\setminus\{1\}\right\}.
	$$
	 Thus, for $c_{\w_0}$, we can construct at least $(q/s)^{m-1}$ mutually disjoint recovering sets.
\end{proof} 
\begin{theorem}[Non-binary PIR codes]\label{th::asymptotic non-binary disjoint repair group code}
	Given an integer $m\ge 2$, for any real $\varepsilon$ with $0<\varepsilon <\frac{m-1}{m}$ and a power of two $q$, there exists an  $n^\varepsilon$-PIR code of length $N=q^m$ and dimension $n$ over $\Sigma$
	such that the redundancy, $N-n$,   and the alphabet size, $|\Sigma|$, satisfy
	$$
	N-n = O_m\left(n^{(m-1)/m+(\log \lambda_m - m + 1)\epsilon /(m-1)}\right),\quad |\Sigma|=	q^{\Theta_m(q^{m-\epsilon m^2/(m-1)})}.
	$$
	In other words, for $0<\epsilon<1$, the polynomial growth of the minimal redundancy of $n^\epsilon$-PIR codes with dimension $n$ is
	$$
	\log_n\left(r_{|\Sigma|}(n,n^\epsilon)\right)\le \min_{m \ge \lceil 1/ (1-\epsilon) \rceil}\left( \frac{m-1}{m} + \frac{1+\log \lambda_m - m}{m-1}\epsilon \right).
	$$ 
\end{theorem}
\begin{proof}%
Take $s=\Theta_m(q^{1 - \epsilon m/(m-1)})$. For simplicity of notation, we assume that $s$ is a power of two. By Theorem~\ref{th:: nonbinary disjoint repair group code}, there exists a $k$-PIR code with $k=(q/s)^{m-1}=\Theta_m(N^\varepsilon) = \Theta_m(n^{\epsilon})$ over $\F_q^{\binom{s+m-1}{m}}$ of length $N=q^m$ and redundancy at most
	\begin{align*}
	N-n&=O_m\left(q^m s^{-1}(q/s)^{\log\lambda_m - m}\right)\\
	&=O_m\left(q^{\epsilon m/(m-1)+(m-1)}q^{\epsilon m/(m-1) (\log\lambda_m - m)}\right)\\
	&=O_m\left(n^{(m-1)/m+(\log \lambda_m - m + 1)\epsilon /(m-1)}\right) .
	\end{align*}
\end{proof}

We now transform the non-binary codes constructed in Theorem~\ref{th::asymptotic non-binary disjoint repair group code} into binary PIR codes.
\begin{theorem}[Binary PIR codes]\label{th::binary disjoint repair group property code}
	Given a positive integer $m$, for any real $\varepsilon$ with $0<  \varepsilon <\frac{m-1}{m}$, any real $\delta > 0$ and an integer $n$ sufficiently large,
	there exists a binary $n^{\varepsilon-\delta}$-PIR code of length $N$ and dimension $n$
	such that the redundancy, $N-n$, satisfies
	$$
	N-n = O_m\left(n^{(m-1/2)/m + \epsilon (1/2 + \log \lambda_m - m)/(m-1) }\right).
	$$
		In other words, for $0<\epsilon<1$, the polynomial growth of the minimal redundancy of binary $n^\epsilon$-PIR codes with dimension $n$ is
	$$
	\log_n\left(r(n,n^\epsilon)\right)\le \min_{m \ge \lceil 1/ (1-\epsilon) \rceil}\left( \frac{2m-1}{2m} + \frac{1+2\log \lambda_m - 2m}{2m-2} \epsilon \right).
	$$ 
\end{theorem}
\begin{proof}
	Let $\C$ be a non-binary PIR code as in Theorem~\ref{th::asymptotic non-binary disjoint repair group code}. We construct the binary PIR code $\overline{\C}$ from $\C$ by converting each symbol of the alphabet of size $|\Sigma|=q^{\Theta_m(q^{m - \epsilon m^2/(m-1)})}$ to 
	$$
	\log|\Sigma|=\Theta_m(q^{m - \epsilon m^2/(m-1)}\log q)=\Theta_m(N^{1-\epsilon m/(m-1)}\log N)=\Theta_m(n^{1-\epsilon m/(m-1)}\log n)
	$$ 
	bits. Denote the length and dimension of the binary code by $\overline{N}$ and $\overline{n}$, respectively. Thus, $\overline{n}= \Theta_m(n^{2-\epsilon m/(m-1)}\log n)$ and $\overline{N}=\Theta_m(n^{2-\epsilon m/(m-1)}\log n)$. Therefore, $n=\Theta_m({\overline{n}}^{(m-1)/(2m-2-\epsilon m)}/\log \overline{n})$. Denote by $\overline{r} = \overline{N}-\overline{n}=(N-n)\log|\Sigma|$ the redundancy and by $\overline{k}$ the availability parameter of the new code. 
	
	First, we note that the availability parameter of $\overline{\C}$ is at least that of $\C$. Indeed, we know that each bit in $\overline{\C}$ is a bit among $\log |\Sigma|$ bits representing some symbol in $\C$. For each recovering set of a symbol in $\C$, we get a corresponding recovering set for any bit from the image of this symbol in $\overline{\C}$. Therefore, $\overline{k}\ge k=n^\epsilon\ge \Theta_m({\overline{n}}^{\epsilon(m-1)/(2m-2-\epsilon m)}/(\log \overline{n})^\epsilon)$. Define $\overline{\epsilon}= \epsilon(m-1)/(2m-2-\epsilon m)$. Then $\overline{k}=\Omega_m(\overline{N}^{\overline{\epsilon}} / \log \overline{n})$ and $\epsilon = (2m-2)\overline{\epsilon}/(m-1+\overline{\epsilon} m)$
	
	Second, we rewrite the redundancy $\overline{r}$ in terms of $\overline{n}$ and $\overline{\epsilon}$ as
	\begin{align*}
	\overline{r}&=\overline{N}-\overline{n}=O_m\left(n^{(m-1)/m+(\log \lambda_m - m + 1)\epsilon /(m-1)} n^{1-\epsilon m/(m-1)}\log n\right) \\
	&=O_m\left(n^{(2m-1)/m + (\log \lambda_m - 2m + 1)\epsilon / (m-1)}\log n\right) \\
	 	&=O_m\left(\overline{n}^{(m-1)(2m-1)/(2m^2 - 2m - 2\epsilon m^2) + (\log \lambda_m - 2m + 1)\epsilon/(2m-2-\epsilon m) }\log \overline{n} \right) \\
	 	&=O_m\left(\overline{n}^{(m-1/2)/m + \overline{\epsilon} (1/2 + \log \lambda_m - m)/(m-1) } \log \overline{n}\right).
	\end{align*}
	As for any $\delta>0$ and sufficiently large $n$ we have $\log n < n^{\delta}$, the required statement is proved.
\end{proof}
\subsection{LCCs from lifted multiplicity codes} \label{ss::locally correctable codes}

In this section,  we prove that a lifted multiplicity code is a LCC with certain parameters (cf. Definition~\ref{def::LCC}). More specifically, we describe the self-correction algorithm for lifted multiplicity codes. This algorithm is slightly better in terms of locality and running time than the self-correction algorithm presented in~\cite{kopparty2014high}, but we impose a stronger requirement on $s$, the order of derivatives. It is worth mentioning that the algorithm for multiplicity codes from~\cite{kopparty2014high} also works well for lifted multiplicity codes.

One important ingredient for showing the self-correction algorithm is the following statement about hypergraphs. Recall that a \textit{$s$-partite hypergraph} $H$ is a pair $H=(V,E)$, where $V$ is the vertex set that can be partitioned into sets $V_1,\ldots,V_s$ so that each edge in the edge set $E$ consists of a choice of precisely one vertex from each part. By $K_l^{(s)}$ denote a \textit{complete $s$-partite hypergraph}, whose parts are all of equal size $l$.

\begin{theorem}[Follows from~{\cite[Theorem~1]{erdos1964extremal}}] \label{th::Turan's type theorem}
 		Let $n>sl$, $l>1$.  Then every $s$-partite hypergraph with $n$ vertexes and at least $n^{s-1/l^{s-1}}$ hyperedges contains a copy of $K^{(s)}_l$.
 \end{theorem}

\begin{theorem}
Let $m$ be a fixed positive integer.    For $s^{m-2}=o(\log q)$ and a real $\alpha<1/4$, the $[m, s, qs-r, q]$ lifted multiplicity code is a $((q-1)s^{m-1}, \alpha \Delta_{min}, 2\alpha +o(1))$-locally correctable code, where $\Delta_{min}:=\left \lceil\frac{r-s+1}{s}\right\rceil\frac{q-s}{q^2}$.
\end{theorem}
\begin{proof}
    We prove this theorem by presenting a new self-correction algorithm $\A$ for lifted multiplicity codes. Consider a vector $\y=(y_1, \ldots, y_{q^m})=(y_{\w})|_{\w\in \F_q^m}$, which is a noisy version of the evaluation of the polynomial $f$. Say that we want to correct the value $f^{(<s)}(\w_0)$.  The algorithm $\A$ consists of three steps. \medskip
    
        \textbf{Step 1:}
        Choose sets $Q_2, Q_3, \ldots, Q_m$, $Q_i\subset \F_q$, $|Q_i|=s$, independently according to the uniform distribution over all subsets of size $s$. Form a set $V$ of directions $\v=(1, v_2, \ldots, v_m), v_i\in Q_i$. \medskip
        
       \textbf{Step 2:}
        For every $\v \in V$ define a polynomial $g_{\v}(T):=f(\w_0+T\v)$. By the definition of lifted multiplicity codes we know that this polynomial agrees with some uni-variate polynomial of degree less than $qs-r$ on its first  $s-1$ derivatives. Apply the decoding algorithm for a uni-variate multiplicity code from~\cite{kopparty2014high, sudan2001ideal} to noisy evaluations of $g_{\v}(T)$ to obtain an estimation $\hat{g}_{\v}(T)$ of the correct polynomial $g_{\v}(T)$. Note that this decoding algorithm can correct up to $\lfloor (d_{min}-1)/2 \rfloor$ errors, where $d_{min}:=\lceil\frac{r+1}{s}\rceil$.\medskip
        
       \textbf{Step 3:}
        Using Lemma~\ref{lem::good direction and multivariate polynomial} and polynomials $\hat{g}_{\v}(T)$, recover the value $f^{(<s)}(\w_0)$ to obtain $\hat f^{(<s)}(\w_0)$.\medskip
        
    We now present an analysis of the algorithm. Call a direction $\v$ \textit{good}, if the line $\w_0+T\v$ contains at most $\lfloor (d_{min}-1)/2 \rfloor$ errors. Note that if a direction $\v$ is good, then $\hat{g}_{\v}(T)\equiv_s g_{\v}(T)$. Thus, if all directions from $V$ are good, the algorithm  recovers the symbol correctly, i.e., $\hat f^{(<s)}(\w_0)=f^{(<s)}(\w_0)$. In the following we derive a bound on the probability that all directions from $V$ are good.
    
    Introduce an $(m-1)$-uniform $(m-1)$-partite hypergraph $H$, each part of which has size $q$. Index the elements within each part of the hypergraph with elements of $\F_q$. For every good direction $\v=(1, v_2, \ldots, v_{m})$, draw a hyperedge $(v_2, \ldots, v_{m})$ in $H$, where $v_i$ is a vertex from the $(i-1)$th part. 
    Then the probability of the successful recovery of $f^{(<s)}(\w_0)$ is lower bounded by the number of copies of $K^{(m-1)}_s$ in $H$ divided by $q^{m-1}$.
    
    The total number of good directions (or hyperedges in $H$) is at least 
    $$
    q^{m-1}- \frac{\alpha \Delta_{min}q^m}{\lfloor (d_{min}-1)/2 \rfloor}=q^{m-1}(1-2\alpha +o(1)).
    $$
    
     We show how we can find a large number of copies of $K^{(m-1)}_s$ in $H$.
 As long as the number of hyperedges in $H$ is greater than $((m-1)q)^{m-1-1/s^{m-2}}$ we can find such a copy by Theorem~\ref{th::Turan's type theorem}. Then, we can spoil this copy by erasing one of its hyperedges and repeat the process for the obtained hypergraph. Obviously, all constructed copies of $K^{(m-1)}_s$ would be distinct. By this procedure, we can find at least $$q^{m-1}(1-2\alpha+o(1))-((m-1)q)^{m-1-1/s^{m-2}}=q^{m-1}(1-2\alpha+o(1))$$ copies of $K^{(m-1)}_s$. Therefore, the probability of successful decoding is at least $1-2\alpha+o(1)$.

\end{proof}
\section{Conclusion}\label{ss::conclusion}
In this paper, we have investigated the rate, the distance, the availability and the self-correction properties of lifted multiplicity codes based on the evaluations of $m$-variate polynomials and discussed how to use them to construct PIR codes and LCCs. For some parameter regimes, lifted multiplicity codes are shown to have a better 
rate/distance/availability/locality trade-off than other known constructions. It would be interesting to see  whether this class of codes can also be of use for other applications and settings.

\section{Appendix}
\subsection{Proof of Proposition~\ref{prop::code cardinality}}\label{ss::injective map}
The proof is twofold, we need to show that
\begin{enumerate}[wide=\parindent]%
    \item[\textit{(Distinction)}] the evaluation of every monomial $\X^{\d}$ with $\deg_q(\d)\leq s-1$, which we refer to as a \textit{type-$s$} monomial, gives a unique word
    \item[\textit{(Inclusion)}] these words are contained in the $[m,s,d,q]$ lifted multiplicity code as in Defintion~\ref{def::lifted mult code} .
\end{enumerate}

To show that the words are distinct, it is sufficient to prove that for an arbitrary non-trivial linear combination, written as $f(\X)$, of type-$s$ monomials, its evaluation is not equal to the all-zero codeword. Our proof is a straightforward generalization of~\cite[Lemma 14]{li2019lifted}. 
	
	We prove the proposition by induction on $m$ and $s$. More precisely, we deduce the statement for $(m, s)$ from the cases for $(m-1, s)$ and $(m, s-1)$. The base case $m=1$ is equivalent to~\cite[Lemma~11]{li2019lifted}. In the base case $s=1$ the degree of each variable in $f$ is at most $q-1$. Then the proposition follows from DeMillo–Lipton–Zippel Theorem~\cite{demillo1977probabilistic,zippel1979probabilistic}, which states that such polynomial can't have more than $q^m-(q-(q-1))^m=q^m-1$ zeroes.
	
	Now we prove the inductive step. Assume that $f(\X)$ is a non-trivial linear combination of type-$s$ monomials such that $f(\X)\equiv_s 0$. Consider the polynomial $g(X_1,\ldots,X_{m-1}):=f(X_1,\ldots, X_{m-1}, c)$ in $m-1$ variables, where $c\in \F_q$ is fixed. By the inductive hypothesis, we conclude that $g\equiv_s 0$. Hence, $(X_m-c)$ divides $f(\X)$ for all $c\in \F_q$, so $(X_m^q-X_m)$ divides $f(\X)$. Therefore, $f(\X)$ can be represented as $f(\X)=(X_m^q-X_m)g(\X)$.
	
	It is easy to see that $g(\X)$ is a linear span of type-$(s-1)$ monomials. Taking the $\i$th derivative of $f(\X)$ for any $\i\in\Z_{\ge}^m$ with $i_m\geq 1$ we obtain
	$$
	f^{(\i)}(\X)=(X_m^q-X_m)g^{(\i)}(\X)-g^{(\j)}(\X),
	$$
	where $\j=(i_1,\ldots, i_{m-1}, i_m-1)$. The left-hand side is equal to zero for all $\x\in \F_q^m$ and $\i\in\Z_{\ge}^m$ with $\deg(\i)\le s-1$. The right-hand side equals to $-g^{(\j)}(\x)$ for all $\x\in \F_q^m$ and all $\j\in\Z_{\ge}^m$ with $\sum\limits_{l=1}^{m-1} j_l<s-1$. By the induction hypothesis $g(\X)$ is the zero polynomial, thus, $f(\X)$ is the zero polynomial as well. This concludes the proof of the distinction property.

To show the inclusion, we prove that every $(d,s)$-good monomial $f(\X)=\X^\d$ over $\F_q$ satisfies the property that for any line $L\in\L_m$, the restriction $f|_L$ is equivalent up to order $s$ to an uni-variate polynomial of degree less than $d$. Let a line $L$ be parameterized as $(\a T+\b)|_{T\in\F_q}$ and $\0$ be the all-zero vector. Then, we have that 
\begin{align*}
f|_L &=(\a T+\b)^\d \\
&= \sum_{\0\le \i\le \d}\prod_{j=1}^{m}a_j^{i_j}b_j^{d_j-i_j}\binom{d_j}{i_j} T^{i_j}\\
&\equiv_s \sum_{k=0}^{qs-1} c_k T^k := f^*(T),
\end{align*}
where $c_k$ denotes the coefficients of the unique polynomial of degree $\leq qs-1$ that is equivalent to $f|_L$ (cf. Proposition~\ref{pr::reducing the power}). Recall that $s$ and $q$ are powers of $2$. Hence, we have $f|_L(T) = f^*(T) \pmod{T^{qs} + T^s}$ by Proposition~\ref{pr::reducing the power}, so the coefficients $[T^s]f|_L$ that contribute to the coefficient $c_k$ are exactly those for which $s=\deg(\i) \Modsp{q}{s} = k$, and we obtain
\begin{equation}\label{eq::coefficients equivalent polynomial}
c_k\eqdef \sum_{\substack{\0\le \i\le \d \\ \deg(\i)\Modsp{q}{s} = k }}\prod_{j=1}^{m}a_j^{i_j}b_j^{d_j-i_j}\binom{d_j}{i_j}.
\end{equation}
By Definition~\ref{def::bad (d^*,s) monomial}, for $k\ge d$, there is no $\i\in\Z_{qs}^m$ such that $\i\le_2\d$ and $\deg(\i)\Modsp{q}{s} = k$. Thus, for $k\ge d$ and every $\i$ used in the summation of \eqref{eq::coefficients equivalent polynomial}, there exists some coordinate $j\in[m]$ such that $i_j\not\le_2 d_j$. By Lucas's Theorem (e.g., see~\cite{guo2013new,li2019lifted}), for integers $d_j=(d_j^{(\ell-1)},...,d_j^{(0)})_2$ and $i_j=(i_j^{(\ell-1)},...,i_j^{(0)})_2$ it holds that 
\begin{equation*}
    \binom{d_j}{i_j} = \prod_{\xi=0}^{\ell-1} \binom{d_j^{(\xi)}}{i_j^{(\xi)}} \mod 2  .
\end{equation*}
It follows that if $i_j\not\le_2 d_j$ the coefficient $\binom{d_j}{i_j}=0$ in $\F_q$ (as $q$ is a power of two) and therefore $c_k=0$ for all $k\ge d$.

We have proved that the restriction of $\X^\d$ to any line is an uni-variate polynomial of degree at most $d-1$. Therefore, the $[m,s,d,q]$ lifted multiplicity code includes the codewords 
\begin{equation*}
\{(\X^\d)|_{\a\in\F_q^m}\ : \ \X^\d \in \mathcal{F}_{q,s,d} \} \ . %
\end{equation*}
The inclusion of their linear combinations over $\F_q$ follows trivially from the proof.
\qed

\subsection{Lifted multiplicity code and lifted multiplicity monomial code}\label{ss::equivalenceLiftedRS}

We now give an example showing that lifted multiplicity codes are not necessarily spanned by the set of good monomials. %
 Let $d=qs-2$, $s=2$, and $q>2$. 
 Denote by $M_1(\X)$ and $M_2(\X)$ the monomials
\begin{align*}
  M_1(\X) &:= \X^{\d^{(1)}} = X_1^{qs-2}X_2\\
  M_2(\X) &:= \X^{\d^{(2)}} = X_1^{(s-1)q-1}X_2^q\ , 
\end{align*}
so $d_{1}^{(i)} = qs-2$, $d_{2}^{(1)} = 1$, $d_{1}^{(2)}=(s-1)q-1$, and $d_{2}^{(2)}=q$.
Both monomials are type-$s$ as
\begin{align*}
 \deg_q(\d^{(1)}) = \deg_q(\d^{(2)}) = qs-1 < 2(q-1) + (s-1)q
\end{align*}
Further, both are $(d,s)$-bad, as the vectors $\i^{(1)} = \d^{(1)}$ and $\i^{(2)} = \d^{(2)}$ fulfill Definition~\ref{def::bad (d^*,s) monomial} for each monomial, respectively. Also, their evaluation is not contained in an $[m,s,d,q]$ lifted multiplicity code, since for the line $(0,w_2)+(1,v_2)T \in \L_2$ we have
\begin{align*}
  [T^{qs-1}] M_1(T,w_2+v_2 T) &= v_2\\
  [T^{qs-1}] M_2(T,w_2+v_2 T) &= v_2^q \ .
\end{align*}
However, the evaluation of their sum, i.e., the polynomial
\begin{align*}
  P(\X) := M_{1}(\X) + M_{2}(\X) \ ,
\end{align*}
is contained in the $[m,s,d,q]$ lifted multiplicity code as
\begin{align*}
  [T^{qs-1}] P(w_1 + v_1T,w_2+v_2 T) &= [T^{qs-1}]  M_1(w_1 + v_1T,w_2+v_2 T)+ [T^{qs-1}]  M_2(w_1 + v_1 T,w_2+v_2 T) \\
  &= v_1^{qs-2}v_2 + \underbrace{v_1^{(s-1)q-1}v_2^q}_{\stackrel{\mathsf{(a)}}{=} v_1^{qs-2}v_2} = 0 \ ,
\end{align*}
where $\mathsf{(a)}$ holds because $v_1, v_2 \in \F_q$.

\subsection{Multiplicity codes vs. lifted multiplicity codes} \label{ss::improvement LMC vs MC}

To provide some intuition and show how lifting can improve the rate of multiplicity codes, we give an example for a fixed set of parameters here. 
Let $m=s=2$, $q=4$, and $d=qs-1=7$. Consider the monomial $M(\X) := X_1^2 X_2^6$. The degree of this monomial is $\deg(M(\X))=8>d$, so its evaluation is not contained in the $[2,2,7,4]$ multiplicity code, as it only contains evaluations of degree $<d$ polynomials.

By Definition~\ref{def::lifted mult code}, the evaluation of $M(\X)$ is contained in the $[2,2,7,4]$ lifted multiplicity code if for every line $L\in \mathcal{L}_m$ there exists a polynomial $g(T)\in \mathcal{F}_{q,d}$ such that the restriction of $M(\X)$ to $L$ is equivalent to $g(T)$. First, note that $M(\X)$ is a type-$s$ monomial, as $\deg_q(M(\X)) = 1 \leq s-1$. Its evaluation in an arbitrary line $L\in\mathcal{L}_2$ is given by
\begin{align*}
    M(\X)|_L &= (w_1 + v_1T)^2 (w_2+v_2T)^6\\
    &= (w_1^2 + v_1^2T^2)(w_2^6+w_2^4v_2^2 T^2 + w_2^2v_2^4 T^4 + v_2^6T^6) \\
    &= w_1^2w_2^6 + (w_1^2w_2^4v_2^2 + v_1^2w_2^6)T^2 + (w_1^2w_2^2v_2^4 + v_1^2w_2^4v_2^2)T^4 + (w_1^2v_2^6+v_1^2 w_2^2v_2^4)T^6 + v_1^2v_2^6T^8 \ .
\end{align*}
By Proposition~\ref{pr::reducing the power} and because $s$ and $q$ are powers of $2$, we know that there exists an equivalent polynomial $M^*(T)$ of degree at most $qs-1=7$ such that $M(\X)|_L \equiv_s M^*(T) \pmod{T^{8}+T^2}$. Here, we obtain this polynomial by substracting $v_1^2v_2^6(T^{8}+T^2)$ from $M(\X)|_L$, which gives
\begin{align*}
    M^*(T) =w_1^2w_2^6 + (w_1^2w_2^4v_2^2 + v_1^2w_2^6+v_1^2v_2^6)T^2 + (w_1^2w_2^2v_2^4 + v_1^2w_2^4v_2^2)T^4 + (w_1^2v_2^6+v_1^2 w_2^2v_2^4)T^6 \ .
\end{align*}
As the degree of this polynomial is $\deg(M^*(T))<d=7$ its evaluation is contained in the $[2,2,7,4]$ lifted multiplicity code, thereby increasing its dimension compared to the $[2,2,7,4]$ multiplicity code.

\bibliographystyle{IEEEtran}
\bibliography{lifted}

\begin{thebibliography}{10}
\providecommand{\url}[1]{#1}
\csname url@samestyle\endcsname
\providecommand{\newblock}{\relax}
\providecommand{\bibinfo}[2]{#2}
\providecommand{\BIBentrySTDinterwordspacing}{\spaceskip=0pt\relax}
\providecommand{\BIBentryALTinterwordstretchfactor}{4}
\providecommand{\BIBentryALTinterwordspacing}{\spaceskip=\fontdimen2\font plus
\BIBentryALTinterwordstretchfactor\fontdimen3\font minus
  \fontdimen4\font\relax}
\providecommand{\BIBforeignlanguage}[2]{{%
\expandafter\ifx\csname l@#1\endcsname\relax
\typeout{** WARNING: IEEEtran.bst: No hyphenation pattern has been}%
\typeout{** loaded for the language `#1'. Using the pattern for}%
\typeout{** the default language instead.}%
\else
\language=\csname l@#1\endcsname
\fi
#2}}
\providecommand{\BIBdecl}{\relax}
\BIBdecl

\bibitem{huang2013pyramid}
C.~Huang, M.~Chen, and J.~Li, ``Pyramid codes: Flexible schemes to trade space
  for access efficiency in reliable data storage systems,'' \emph{ACM Trans.
  Storage}, vol.~9, no.~1, pp. 1--28, 2013.

\bibitem{GHSY12}
P.~Gopalan, C.~Huang, H.~Simitci, and S.~Yekhanin, ``On the locality of
  codeword symbols,'' \emph{IEEE Trans. Inf. Theor.}, vol.~58, no.~11, p.
  6925–6934, Nov. 2012.

\bibitem{katz2000efficiency}
J.~Katz and L.~Trevisan, ``On the efficiency of local decoding procedures for
  error-correcting codes,'' in \emph{Proc. 32nd Annu. ACM Symp. Theory Comput.
  (STOC)}, 2000, pp. 80--86.

\bibitem{yekhanin2012locally}
S.~Yekhanin \emph{et~al.}, ``Locally decodable codes,'' \emph{Found. Trends
  Theor. Comput. Sci.}, vol.~6, no.~3, pp. 139--255, 2012.

\bibitem{gur2018relaxed}
T.~Gur, G.~Ramnarayan, and R.~D. Rothblum, ``Relaxed locally correctable
  codes,'' in \emph{Proc. 9th Conf. Innov. Theor. Computer Sci. (ITCS)}, 2018,
  p. 27:1–27:11.

\bibitem{ben2006robust}
E.~Ben-Sasson, O.~Goldreich, P.~Harsha, M.~Sudan, and S.~Vadhan, ``Robust
  {PCP}s of proximity, shorter {PCP}s, and applications to coding,'' \emph{SIAM
  J. Comput.}, vol.~36, no.~4, pp. 889--974, 2006.

\bibitem{holzbaur2020lifted}
L.~{Holzbaur}, R.~{Polyanskaya}, N.~{Polyanskii}, and I.~{Vorobyev}, ``Lifted
  reed-solomon codes with application to batch codes,'' in \emph{2020 IEEE Int.
  Symp. Inf. Theory (ISIT)}, 2020, pp. 634--639.

\bibitem{ishai2004batch}
Y.~Ishai, E.~Kushilevitz, R.~Ostrovsky, and A.~Sahai, ``Batch codes and their
  applications,'' in \emph{Proc. 36th Annu. ACM Symp. Theory Comput. (STOC)},
  2004, pp. 262--271.

\bibitem{fazeli2015pir}
A.~Fazeli, A.~Vardy, and E.~Yaakobi, ``Pir with low storage overhead: coding
  instead of replication,'' \emph{arXiv preprint arXiv:1505.06241}, 2015.

\bibitem{li2019lifted}
R.~Li and M.~Wootters, ``Lifted multiplicity codes and the disjoint repair
  group property,'' in \emph{Proc. Approx. Randomiz. Combinat. Optim. Algor.
  Techn. (APPROX/RANDOM)}, vol. 145, 2019, pp. 38:1--38:18.

\bibitem{reed1954class}
I.~Reed, ``A class of multiple-error-correcting codes and the decoding
  scheme,'' \emph{Trans. IRE Prof. Group Inf. Theory}, vol.~4, no.~4, pp.
  38--49, 1954.

\bibitem{arora2003improved}
S.~Arora and M.~Sudan, ``Improved low-degree testing and its applications,''
  \emph{Combinatorica}, vol.~23, no.~3, pp. 365--426, 2003.

\bibitem{alon2005testing}
N.~Alon, T.~Kaufman, M.~Krivelevich, S.~Litsyn, and D.~Ron, ``Testing
  {Reed-Muller} codes,'' \emph{IEEE Trans. Inf. Theory}, vol.~51, no.~11, pp.
  4032--4039, 2005.

\bibitem{rubinfeld1996robust}
R.~Rubinfeld and M.~Sudan, ``Robust characterizations of polynomials with
  applications to program testing,'' \emph{SIAM J. Comput.}, vol.~25, no.~2,
  pp. 252--271, 1996.

\bibitem{guo2013new}
A.~Guo, S.~Kopparty, and M.~Sudan, ``New affine-invariant codes from lifting,''
  in \emph{Proc. 4th Conf. Innov. Theor. Computer Sci. (ITCS)}, 2013, pp.
  529--540.

\bibitem{ben2011symmetric}
E.~Ben-Sasson, G.~Maatouk, A.~Shpilka, and M.~Sudan, ``Symmetric {LDPC} codes
  are not necessarily locally testable,'' in \emph{IEEE 26th Annu. Conf.
  Comput. Complex. (CCC)}, 2011, pp. 55--65.

\bibitem{kopparty2014high}
S.~Kopparty, S.~Saraf, and S.~Yekhanin, ``High-rate codes with sublinear-time
  decoding,'' \emph{J. Assoc. Comput. Mach.}, vol.~61, no.~5, p.~28, 2014.

\bibitem{wu2015revisiting}
L.~Wu, ``Revisiting the multiplicity codes: A new class of high-rate locally
  correctable codes,'' in \emph{Proc. IEEE 53rd Annu. Allerton Conf. Commun.
  Contr. Comput. (Allerton)}, 2015, pp. 509--513.

\bibitem{TB14}
I.~{Tamo} and A.~{Barg}, ``Bounds on locally recoverable codes with multiple
  recovering sets,'' in \emph{Proc. IEEE Int. Symp. Inf. Theory (ISIT)}, 2014,
  pp. 691--695.

\bibitem{rao2016lower}
S.~Rao and A.~Vardy, ``Lower bound on the redundancy of {PIR} codes,''
  \emph{arXiv preprint arXiv:1605.01869}, 2016.

\bibitem{wootters2016linear}
M.~Wootters, ``Linear codes with disjoint repair groups,'' \emph{unpublished
  mansucript, February}, 2016.

\bibitem{VRK17}
M.~{Vajha}, V.~{Ramkumar}, and P.~{Vijay Kumar}, ``Binary, shortened projective
  reed muller codes for coded private inf retrieval,'' in \emph{Proc. IEEE Int.
  Symp. Inf. Theory (ISIT)}, 2017, pp. 2648--2652.

\bibitem{asi2018nearly}
H.~Asi and E.~Yaakobi, ``Nearly optimal constructions of {PIR} and batch
  codes,'' \emph{IEEE Trans. Inf. Theory}, vol.~65, no.~2, pp. 947--964, 2018.

\bibitem{FGW17}
S.~L. Frank{-}Fischer, V.~Guruswami, and M.~Wootters, ``Locality via partially
  lifted codes,'' in \emph{Proc. Approx. Randomiz. Combinat. Optim. Algor.
  Techn. (APPROX/RANDOM)}, vol.~81, 2017, pp. 43:1--43:17.

\bibitem{LC04}
S.~Lin and D.~J. Costello, \emph{Error control coding: fundamentals and
  applications}.\hskip 1em plus 0.5em minus 0.4em\relax Upper Saddle River, NJ:
  Pearson/Prentice Hall, 2004.

\bibitem{polyanskii2019lifted}
N.~Polyanskii and I.~Vorobyev, ``Trivariate lifted codes with disjoint repair
  groups,'' in \emph{Proc. IEEE XVI Int. Symp. Probl. Redund. Inf. Contr. Syst.
  (REDUNDANCY)}, 2019, pp. 64--68.

\bibitem{ben2006LCC}
E.~Ben-Sasson and M.~Sudan, ``Robust locally testable codes and products of
  codes,'' \emph{Random Structures Algorithms}, vol.~28, no.~4, pp. 387--402,
  2006.

\bibitem{viderman2015combination}
M.~Viderman, ``A combination of testability and decodability by tensor
  products,'' \emph{Random Structures Algorithms}, vol.~46, no.~3, pp.
  572--598, 2015.

\bibitem{kopparty2017high}
S.~Kopparty, O.~Meir, N.~Ron-Zewi, and S.~Saraf, ``High-rate locally
  correctable and locally testable codes with sub-polynomial query
  complexity,'' \emph{Journal of the ACM (JACM)}, vol.~64, no.~2, pp. 1--42,
  2017.

\bibitem{trevisan2004some}
L.~Trevisan, ``Some applications of coding theory in computational
  complexity,'' in \emph{Electron. Colloq. Comput. Complex. (ECCC)}, 2004.

\bibitem{kopparty2017local}
S.~Kopparty and S.~Saraf, ``Local testing and decoding of high-rate
  error-correcting codes,'' in \emph{Proc. Electron. Colloq. Comput. Complex.
  (ECCC)}, vol.~24, 2017, p. 126.

\bibitem{dvir2013extensions}
Z.~Dvir, S.~Kopparty, S.~Saraf, and M.~Sudan, ``Extensions to the method of
  multiplicities, with applications to kakeya sets and mergers,'' \emph{SIAM J.
  Comput.}, vol.~42, no.~6, pp. 2305--2328, 2013.

\bibitem{chin1976generalized}
F.~Y. Chin, ``A generalized asymptotic upper bound on fast polynomial
  evaluation and interpolation,'' \emph{SIAM J. Comput.}, vol.~5, no.~4, pp.
  682--690, 1976.

\bibitem{erdos1964extremal}
P.~Erd{\"o}s, ``On extremal problems of graphs and generalized graphs,''
  \emph{Israel J. Math.}, vol.~2, no.~3, pp. 183--190, 1964.

\bibitem{sudan2001ideal}
M.~Sudan, ``Ideal error-correcting codes: Unifying algebraic and
  number-theoretic algorithms,'' in \emph{Proc. Int. Symp. Applied Algebra
  Algebr. Algor. Error-Correcting Codes (AAECC)}.\hskip 1em plus 0.5em minus
  0.4em\relax Springer, 2001, pp. 36--45.

\bibitem{demillo1977probabilistic}
R.~A. DeMillo and R.~J. Lipton, ``A probabilistic remark on algebraic program
  testing,'' \emph{Inf. Process. Lett.}, vol.~7, no.~4, p. 193–195, 1977.

\bibitem{zippel1979probabilistic}
R.~Zippel, ``Probabilistic algorithms for sparse polynomials,'' in \emph{Proc.
  Int. Symp. Symb. Algebr. Manipul. (SYMSAC)}.\hskip 1em plus 0.5em minus
  0.4em\relax Springer, 1979, pp. 216--226.

\end{thebibliography}
\end{document}